\documentclass[sigconf]{acmart}
\AtBeginDocument{%
  }

\setcopyright{acmlicensed}
\copyrightyear{2018}
\acmYear{2018}
\acmDOI{XXXXXXX.XXXXXXX}
\acmConference[Conference acronym 'XX]{Make sure to enter the correct
  conference title from your rights confirmation email}{June 03--05,
  2018}{Woodstock, NY}
\acmISBN{978-1-4503-XXXX-X/2018/06}

\usepackage{seqsplit}

\usepackage[symbol]{footmisc}

\usepackage{amsmath}

\usepackage{algorithm}
\usepackage{algpseudocode}

\usepackage{amsthm}
\newtheorem{theorem}{Theorem}

\newtheorem{definition}{Definition}
\newtheorem{remark}{Remark}

\usepackage{array}
\usepackage{booktabs}
\usepackage{arydshln}
\newcommand{\ra}[1]{\renewcommand{\arraystretch}{#1}}
\usepackage{threeparttable} 

\usepackage{verbatim}

\usepackage{graphicx}

\usepackage[caption=false,font=footnotesize]{subfig}

\usepackage{enumitem}

\begin{document}

\title{New Permutation Decomposition Techniques For Efficient Homomorphic Permutation}
\author{Xirong Ma}
\authornote{Both authors contributed equally to this research.}
\affiliation{%
  \institution{Shandong University}
  \city{Jinan}
  \country{China}
}
\email{xirongma@mail.sdu.edu.cn}

\author{Junling Fang}
\authornotemark[1]
\affiliation{%
  \institution{Shandong University}
  \city{Jinan}
  \country{China}
}
\email{jlfang@mail.sdu.edu.cn}

\author{Chunpeng Ge}
\authornote{Corresponding author.}
\affiliation{%
\institution{Shandong University}
    \city{Jinan}
  \country{China}
}
\email{chunpeng_ge@sdu.edu.cn}

\author{Dung H. Duong}
\affiliation{%
\institution{University of Wollongong}
  \city{Wollongong}
  \country{Australia}
}
\email{hduong@uow.edu.au}

\author{Yali Jiang}
\affiliation{%
   \institution{Shandong University}
   \city{Jinan}
  \country{China}
}
\email{jiang.yl@sdu.edu.cn}

\author{Yanbin Li}
\affiliation{%
\institution{Shandong University}
    \city{Jinan}
  \country{China}
}
\email{yanbinli@sdu.edu.cn}

\author{Willy Susilo}
\affiliation{%
\institution{University of Wollongong}
\city{Wollongong}
\country{Australia}
}
\email{wsusilo@uow.edu.au}

\author{Lizhen Cui}
\affiliation{%
\institution{Shandong University}
    \city{Jinan}
  \country{China}
}
\email{clz@sdu.edu.cn}

\renewcommand{\shortauthors}{Xirong Ma et al.}

\begin{abstract}
Homomorphic permutation is fundamental to privacy-preserving computations based on batch-encoding homomorphic encryption. It underpins nearly all homomorphic matrix operations and predominantly influences their complexity. Permutation decomposition as a potential approach to optimize this critical component remains underexplored. \textcolor{black}{In this paper, we propose novel decomposition techniques to optimize homomorphic permutations, advancing homomorphic encryption-based privacy-preserving computations.} 

\textcolor{black}{We start by defining an ideal decomposition form for permutations and propose an algorithm searching for depth-1 ideal decompositions. }
\textcolor{black}{Based on this, we prove} the full-depth ideal decomposability of permutations used in specific homomorphic matrix transposition (HMT) and multiplication (HMM) algorithms, allowing them to achieve asymptotic improvement in speed and rotation key reduction. 
As a demonstration of applicability, substituting the HMM components in the best-known inference framework of encrypted neural networks with our enhanced version shows up to a \textcolor{black}{$3.9\times$} reduction in latency.

\textcolor{black}{We further devise a new method for computing arbitrary homomorphic permutations, specifically those with weak structures that cannot be ideally decomposed. 
We design a network structure that deviates from the conventional scope of decomposition and outperforms the state-of-the-art technique with a speed-up of up to \(1.69 \times\) under a minimal rotation key requirement.}


\end{abstract}

\ccsdesc[500]{Security and privacy~Cryptography}

\keywords{Homomorphic Encryption, Homomorphic Permutation, Permutation Decomposition, Homomorphic Matrix Operation}

\maketitle


\section{Introduction}\label{sec: introduction}

Homomorphic encryption (HE) is a promising cryptographic primitive that facilitates computations on encrypted data without revealing the underlying plaintext. Significant technical advances in HE have been made following Gentry's breakthrough \cite{gentry2009fully}. Among these is a class of HE schemes supporting SIMD (single instruction, multiple data) techniques \cite{gentry2012fully,smart2014fully}. These schemes \cite{brakerski2011fully,fan2012somewhat,brakerski2014leveled,halevi2019improved,cheon2017homomorphic} allow multiple plaintext values to be packed as \textcolor{black}{a vector (or matrix)} and encrypted into a single ciphertext, enabling homomorphic coordinate-wise additions, multiplications, and collective shifts (known as rotations) of vector components. The advantages of batch encoding and parallel computation result in good amortized complexity per plaintext value, making them applicable in various privacy-preserving computation scenarios. 

In HE-based privacy-preserving computations, particularly those related to machine learning, matrix operations on encrypted data are required as core computing components. For instance, homomorphic matrix multiplication and transposition are indispensable for the homomorphic evaluation of neural network propagation algorithms \cite{gilad2016cryptonets,jiang2018secure,lee2022privacy,sav2021poseidon,sav2022privacy,brutzkus2019low}.

In the plaintext computational scenario, matrix operations need to retrieve and arrange matrix entries frequently by index addressing. 
However, for matrices encrypted under a batch-encoding HE scheme,
this task can only be achieved through basic ciphertext operations like addition, multiplication, and rotation. We refer to this process as \textit{homomorphic permutation}.

Homomorphic permutation often dominates the complexity of homomorphic matrix operations, making its optimization critical for improving the performance of HE-based privacy-preserving applications. Specifically, homomorphic permutation necessitates a substantial number of ciphertext rotations \cite{halevi2014algorithms,halevi2018faster} which are among the most computationally intensive operations in HE, involving multiple number-theoretic transformations over the polynomial ring. In addition, distinct rotation keys are required to achieve rotations of varying distances, each of which can be significantly larger than the ciphertext itself. The need for numerous rotation keys introduces considerable overhead in both communication and storage \cite{lee2023rotation,chan2024ark}.

One common optimization for homomorphic permutation involves deconstructing ciphertext rotations and reordering subroutine sequences to balance time complexity with the number of rotation keys. This includes the \textit{hoisting} technique proposed by Halevi and Shoup \cite{halevi2018faster} and its extended version known as \textit{double hoisting} for homomorphic linear transformation \cite{bossuat2021efficient}. 

\textcolor{black}{Another compatible optimization is \textit{permutation decomposition}}. It aims to decompose a permutation into multiple simpler ones, each requiring fewer rotations and keys. This notion was first introduced by Gentry \textit{et al.} \cite{gentry2012fully} where they used the Benes network \cite{benevs1964optimal} for general-purpose permutation decomposition. 
\textcolor{black}{However, its adoption in privacy-preserving computation remains limited \cite{meftah2021doren}, as the complexity bound provided by existing general decomposition techniques \cite{gentry2012fully,halevi2014algorithms} is only optimized in the permutation length\footnote[2]{\textcolor{black}{Throughout this paper, we consistently refer to the length of a permutation as the dimension of the vector upon which the permutation operates.}}. This bound tends to be loose and insufficiently adaptable to permutations with real-world usages \cite{jiang2018secure,rizomiliotis2022matrix,meftah2022towards,MA2024103658}, where the computational complexity is more precisely determined by the number and distribution of diagonals containing non-zero entries in the matrix representation of the permutations \cite{halevi2014algorithms,jiang2018secure}.}

In this paper, we introduce new techniques to devise more effective decompositions, achieving enhancement in both specific and general use cases of homomorphic permutation.

\subsection{Technique Overview \& Contributions}

We begin by analyzing the potential optimal enhancements that decompositions can bring to homomorphic permutations. \textcolor{black}{This motivates us to define an ideal decomposition form, which adaptively reduces the complexity of the homomorphic permutation from polynomial (square root) to logarithmic in the number of non-zero diagonals of the permutation's matrix representation.}
The rest of the paper revolves around achieving or closely approximating the defined ideal optimization for specific and arbitrary permutations:

Firstly, we present an algorithm that searches depth-1 ideal decomposition solutions \textcolor{black}{for arbitrary permutations} (Section \ref{sec: singleround}). This enables us to discover the full-depth ideal decomposability of specific permutations in homomorphic matrix transposition (HMT) proposed by Jiang \textit{et al.}\cite{jiang2018secure} and homomorphic matrix multiplication (HMM) by Rizomiliotis and Triakosia \cite{rizomiliotis2022matrix} (Section \ref{sec: DMP4HMTHMM}), allowing these privacy-preserving computations to achieve asymptotic improvement in speed and rotation key reduction. We also present the application of our enhanced HMM, which performs inference on an encrypted convolutional neural network model. By substituting the HMM component in the state-of-the-art solution with our version, we achieve a significant improvement in latency. 
More usage of the search algorithm for decomposing other valuable permutations are provided in the appendix (Appendix \ref{app: jiangdmp}). 

Secondly, since not every permutation guarantees ideal decomposability, 
we further devise a new algorithm for decomposing arbitrary homomorphic permutations (Section \ref{sec: HVP}). This aims to better approximate the performance of ideal decomposition than previous approaches. We construct a multi-group network structure with nodes taking vectors as input and output, applying rotations with varying steps. Through precise node transmission logic and tailored homomorphic operations, our approach outperforms the state-of-the-art implementation using the Benes network \cite{gentry2012fully,halevi2014algorithms} under minimal rotation key requirements.

\section{Background and Preliminary}\label{sec: background}

\subsection{Notations}\label{sec: notations}
\textcolor{black}{Scalars and polynomials are denoted by lowercase letters (e.g., $a$), vectors by bold lowercase letters (e.g., $\mathbf{a}$), and matrices by uppercase letters (e.g., $A$).} 
$A[i,j]$ represents the element in the $i$-th row and $j$-th column of matrix $A$. We also use $A(k,l)^*$ to indicate the element in the $k$-th diagonal and $l$-th row of $A$, \textcolor{black}{where $A(k,l)^*=A[l,k+l]$.} 
Index computations are always implicitly reduced by the matrix's dimension. In addition, we denote 
element-wise product by $\odot$.

\subsection{Batch-encoding Homomorphic Encryption}

As mentioned above, the batch-encoding homomorphic encryption is a class of HE schemes \cite{brakerski2011fully,gentry2012fully,fan2012somewhat,cheon2017homomorphic} based on the security assumption of RLWE \cite{lyubashevsky2010ideal,lyubashevsky2013toolkit}, which supports vector spaces as plaintext spaces and enables homomorphic element-wise operations.  

Typically, for \textcolor{black}{an $n$-dimensional vector} $\mathbf{z}$ to be encrypted, a batch-encoding HE scheme first maps it to some polynomial ring space $R = \mathbb{Z}[X]/(X^N + 1)$ \textcolor{black}{where $N$ is a power of two and $n \leq N$.} The mapped plaintext polynomial, denoted $pt(\mathbf{z})$, is then encrypted into a ciphertext $ct(\mathbf{z})$ in $R_{Q_\ell}^2 = (\mathbb{Z}_{Q_\ell}[X]/(X^N+1))^2$ where $Q_\ell = \prod_{i=0}^{\ell} q_i, 0\leq \ell< L$ for some pre-determined modulus chain $Q=[q_0, \dots, q_{L-1}]$. Such ciphertext generally supports the following homomorphic operations (see Appendix \ref{app: HE} for more details): 
\begin{itemize}[label={$\bullet$},leftmargin=1.1em]
    \item $\texttt{Add}(ct(\mathbf{z}_1), ct(\mathbf{z}_2))$: For the input ciphertexts $ct(\mathbf{z}_1) \in R_{Q_\ell}^2$ and $ct(\mathbf{z}_2) \in R_{Q_{\ell'}}^2$, output $ct(\mathbf{z}_1) + ct(\mathbf{z}_2) \bmod{Q_{\min(\ell, \ell')}}$, which is a valid ciphertext of $\mathbf{z}_1 + \mathbf{z}_2$.
    
    \item $\texttt{CMult}(ct(\mathbf{z}_1), pt(\mathbf{z}_2))$: For the input ciphertext $ct(\mathbf{z}_1) \in R_{Q_\ell}^2$ and plaintext $pt(\mathbf{z}_2) \in R_{Q_{\ell'}}$, output $ct(\mathbf{z}_1) \cdot pt(\mathbf{z}_2) \bmod{Q_{\min(\ell, \ell')}}$ as a valid ciphertext of $\mathbf{z}_1 \odot \mathbf{z}_2$.
    
    \item $\texttt{Mult}(ct(\mathbf{z}_1), ct(\mathbf{z}_2), rlk)$: For the input ciphertexts $ct(\mathbf{z}_1) \in R_{Q_\ell}^2$ and $ct(\mathbf{z}_2) \in R_{Q_{\ell'}}^2$, compute the tensor product \textcolor{black}{$ct(\mathbf{z}_1) \otimes ct(\mathbf{z}_2) \bmod{Q_{\min(\ell, \ell')}}$} in $R_{Q_{\min(\ell, \ell')}}^3$. Then switch it back to $R_{Q_{\min(\ell, \ell')}}^2$ using a relinearization key $rlk$ and output it as a valid ciphertext of $\mathbf{z}_1 \odot \mathbf{z}_2$.
    
    \item $\texttt{Rescale}(ct)$: For the input ciphertext $ct \in R_{Q_\ell}^2$, \texttt{Rescale} reduces the scaling factor of the underlying plaintext by \textcolor{black}{$\frac{Q_{\ell-1}}{Q_{\ell}}$, resulting in the output: $\lfloor \frac{Q_{\ell-1}}{Q_{\ell}} \cdot ct \rceil \bmod{Q_{\ell-1}} \in R_{Q_{\ell-1}}^2$}. This operation is required after a $\texttt{Mult}$ or $\texttt{CMult}$ for noise control.
    \item $\texttt{Rot}(ct(\mathbf{z}), k)$: \textcolor{black}{Homomorphically left-shifting each component of $\mathbf{z}$ by $k \bmod{n}$ positions with $n$ being the dimension of $\mathbf{z}$. A rotation key $rtk_k$ is required as an implicit input.}
    
\end{itemize}
We summarize some common constraints of batch-encoding homomorphic encryption schemes: 
\begin{itemize}[label={$\bullet$},leftmargin=1.1em]
    \item Batch-encoding HE only natively supports homomorphic operations with a limited multiplication depth of $L-1$, as $Q_L$ provides at most $L-1$ moduli \textcolor{black}{for} \texttt{Rescale} operations. 
    
    \item The complexity order of the routines is: $\texttt{Mult} \approx \texttt{Rot} > \texttt{Rescale} > \texttt{CMult} > \texttt{Add}$. Here, $\texttt{Add}$ and $\texttt{CMult}$ are typically performed in the NTT domain, with an asymptotic complexity of $O(N)$. While others require NTT-domain switches with an $O(N\log{N})$ complexity. 
    
    \item A rotation key consists of multiple ciphertexts encrypting slices of some Galois automorphism of $s$ scaled by some gadget vector and modulus chain $P$ \cite{brakerski2014leveled}. A Large number of rotation keys may cause a great burden on key generation and a high communication overhead for transmitting the keys \cite{lee2023rotation,chan2024ark}. This can be a significant efficiency factor in some multi-party computation scenarios \cite{aloufi2021computing,mouchet2023multiparty,kim2023asymptotically}.  
\end{itemize}


\subsection{Homomorphic Permutation}\label{sec: introHP}

Given a ciphertext $ ct(\mathbf{x}) $ encrypting an $ n $-dimensional plaintext vector $ \mathbf{x} $ and a length-$n$ permutation $ p $ to be applied to $ \mathbf{x} $, the homomorphic permutation $ ct(p(\mathbf{x})) $ can be represented as a homomorphic linear transformation (HLT) $ ct(U\mathbf{x}) $, where $ U $ is the matrix representation of $p$. \textcolor{black}{Specifically, $U[i,j]=1$ if $p$ maps the $i$-th entry of $\mathbf{x}$ to position $j$, otherwise $U[i,j]=0$.} 
$ ct(U\mathbf{x}) $ is computed by:
\begin{align}\label{eq: orgLT}
ct(U\mathbf{x})\leftarrow \sum_{i=0}^{n-1} \mathbf{u}_i \odot \texttt{Rot}(ct(\mathbf{x});i),
\end{align}
where $ \mathbf{u}_i $ denotes the $ i $-th diagonal of $ U $ \cite{halevi2014algorithms}. The BSGS (Baby-Step Giant-Step) algorithm reformulates this expression to reduce its complexity \textcolor{black}{from $O(n)$ ciphertext rotations to $O(n_1+n_2)$ with $n_1n_2 = n$ as follows}:
\begin{align}\label{eq: BSGS}
\begin{split}
\sum_{i = 0}^{n_2-1} \texttt{Rot}\left(\sum_{j =0}^{n_1-1} \left( \mathbf{u}_{i,j} \odot \texttt{Rot}(ct(\mathbf{x}); j) \right) ; n_1 \cdot i \right).
\end{split}
\end{align}
where $ \mathbf{u}_{i,j} = \texttt{Rot}(\mathbf{u}_{n_1 \cdot i + j}; -n_1 \cdot i) $.
If there exists a common difference $ a $ and range $d \leq \lfloor n/a \rfloor$ such that non-zero diagonals are distributed over an arithmetic sequence $ s = \{ a \cdot i \mid 0 \leq i \leq d \} $, a slight modification of BSGS yields a rotation complexity of $ O(d_1 + d_2) $ with $ d_1 d_2 = d $.  
The hoisting techniques \cite{halevi2018faster,bossuat2021efficient} provide the BSGS algorithm with a trade-off between speed and the number of rotation keys required. Broadly, they require $ O(d_1 + d_2) $ rotation keys and enable the speed of HLT to increase with the BSGS ratio $ d_1/d_2 $ within a certain range (potentially $ 1 \sim 16 $).

\begin{figure*}[t]
    \centering
    \includegraphics[width=1.0\textwidth]{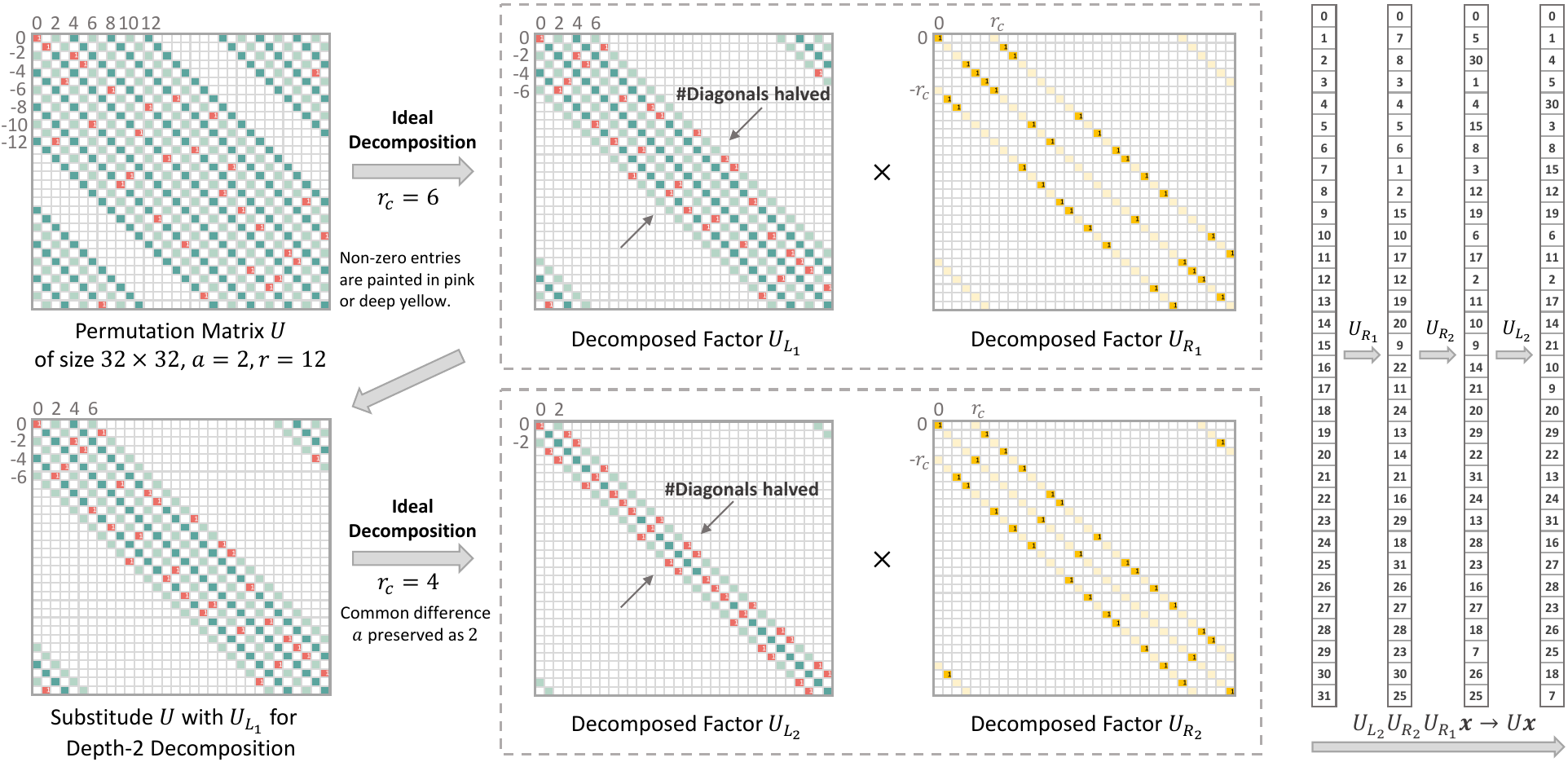}
    \caption{\textcolor{black}{Depth-2 ideal decomposition on a $32\times 32$ permutation with non-zero diagonal indices in $\{2\cdot i\mid -6\leq i\leq 6 \}$}} 
    \label{fig: IdealcmpDemo}
\end{figure*}

\subsubsection{Permutation Decomposition}

Equations \ref{eq: orgLT} and \ref{eq: BSGS} imply that the complexity of homomorphic permutation 
is closely related to the number of non-zero diagonals in its matrix representation. As mentioned in Section \ref{sec: introduction}, permutation decomposition aims to factorize the permutation into factors possessing \textcolor{black}{fewer} non-zero diagonals. This contains two potential benefits: reduced ciphertext rotations and rotation keys.



Gentry \textit{et al.} \cite{gentry2012fully} utilized the Benes network for permutation decomposition in homomorphic permutation. For any permutation $U$ of length $n$, this approach executes $\log{n}$ rounds of decomposition. In the initial iteration, $U$ is factorized into three components: $ U = U^\sigma U^\rho U^\tau $, with $ U^\rho $ serving as the input of the next round. $ U^\sigma $ and $ U^\tau $ contain only non-zero diagonals in $ \{\pm 2^{\log{n}-i-1},0 \} $, where $i$ indicates the iteration round. To our knowledge, the Benes network-based technique \cite{gentry2012fully,halevi2014algorithms} remains the first and optimal general decomposition method for homomorphic permutation. 

Although the Benes network lowers the rotation complexity to roughly $4\log n$, \textcolor{black}{it fails to tailor this cost to the distribution of non‑zero diagonals, which makes it counter‑productive for highly structured permutations. Specifically, it does not consider the arithmetic sequence among the non-zero diagonal indices and may disrupt this regularity during the decomposition rounds. As a result, the expense of evaluating the decomposed factors with HLT is no longer reduced by the length of the arithmetic sequence and can even surpass the cost of executing the original permutation.}
Another limitation of the Benes network lies in its significant depth overhead ($2$ additional levels per round), which can easily exceed the depth supported by homomorphic encryption schemes. While the depth can be reduced by collapsing network levels \cite{halevi2014algorithms},
the complexity of the merged permutations may increase as well.  

\begin{theorem}\label{thm: k=k1+k2}\cite{MA2024103658}
\textcolor{black}{Let $U$ be an $n \times n$ matrix formed by the product of two square matrices $U_L$ and $U_R$. Then, $U[i,j] = \sum_{a=0}^{n-1} U_L[i,a] \cdot U_R[a,j]$, and the diagonal index $k = j - i$ of $U[i,j]$ can be decomposed as the sum of the diagonal index $k_L = a - i$ of $U_L[i,a]$ and the diagonal index $k_R = j - a$ of $U_R[a,j]$, that is, $k = k_L + k_R \mod n$.}

\end{theorem}

Subsequent work designing decompositions for specific permutations gives us new insights \cite{MA2024103658}. It proposed two-factor decompositions in the form $U = U_L U_R$, which includes an intriguing pattern (Theorem \ref{thm: k=k1+k2}). 
This implies that two-factor decompositions can be seen as performing subtraction on the diagonal indices. 
If $ U $ has all its non-zero diagonals distributed within the interval $ [-r,r] \mod{n} $, and we designate its factor, $ U_R $, to possess non-zero diagonals exclusively at \textcolor{black}{position} $ \{r_c, -r_c, 0 \}$ ($r_c\leq r$), then $ U_L $ may have \textcolor{black}{its non-zero diagonals' indices} confined to the range $[-r', r']$ where $r'=\max{(r-r_c,r_c/2)}$. 
Then $ U $ can be substituted by $ U_L $ for further decomposition, as $ U_R $ \textcolor{black}{maintains a small constant number of non-zero diagonals}.
This idea allows $ U $ to be recursively decomposed into a chain:
\begin{equation}\label{eq: idealdmp}
    U = U_{L_\ell} U_{R_\ell} U_{R_{\ell-1}} \dots U_{R_1}.
\end{equation}
The non-zero diagonals of \( U_{L_\ell} \) progressively contract from both sides of the $0$-th diagonal, \textcolor{black}{with their number decreasing as $\ell$ increases.} 

In this paper, we specifically define an ideal decomposition form (Definition \ref{def: idealdmp}). \textcolor{black}{“Ideal” denotes that each round of decomposition halves the number of non‑zero diagonals in the factor decomposed, up to an additive constant of $2$ (excluding the $0$-th diagonal). Figure \ref{fig: IdealcmpDemo} illustrates a depth-$2$ ideal decomposition on a $32\times 32$ permutation matrix whose non-zero diagonal indices follow the arithmetic sequence $s=\{a\cdot i \mid -r/a \leq i \leq r/a\}$ with common difference $a=2$ and distribution range $r=12$. One can observe that the number of non-zero diagonals in $U$ and $U_{L_1}$ are halved in their left factors, while the right factors preserve a constant number of diagonals.}   

\textcolor{black}{A full-depth ideal decomposition requires $O(\log (r/a))$ rotations and a multiplication depth of $\lfloor \log (r/a)\rfloor+1$. Figure \ref{fig: IdealcmpDemo} shows exactly a full-depth decomposition, where the sub-permutations $U_{R_1},U_{R_2},U_{L_2}$ can be sequentially applied to any input ciphertext $ct(\mathbf{x})$ to compute $ct(U\mathbf{x})$, each requiring a single HLT. This results in a total of $2(\lfloor \log{(r/a)} \rfloor+1)=6$ rotations. Furthermore, if the decomposition is partial with a depth of $\ell$, the $U_{L_\ell}$ in the decomposition chain will benefit from a BSGS-based HLT, yielding an overall cost of $O\left(\sqrt{(r/a)/2^\ell}+\ell\right)$, due to the preserved arithmetic sequence in $U_{L_\ell}$.
In contrast to the Benes network, this decomposition scheme decouples the complexity from $n$, implying an enhanced adaptability to structured permutations in which the non-zero diagonals are regularly and sparsely distributed across the permutation matrix.} 
\textcolor{black}{Such an optimization effect may represent a potential upper bound in speed-up achievable through decomposition techniques. }

\begin{definition}\label{def: idealdmp}
    Given a square matrix $ U $ with \textcolor{black}{its non-zero diagonals distributed in interval $[-r,r]$ as an arithmetic sequence $s=\{a\cdot i \mid -r/a \leq i \leq r/a\}$ with the common difference $a\geq 1$,} and a \textcolor{black}{decomposition} $ f $, we say $ f $ is an ideal depth-$\ell$ decomposition
    ($1\leq \ell\leq \lfloor \log{(r/a)}\rfloor$) 
    of $ U $ if \textcolor{black}{$ f(U,\ell) = \{ U_{L_\ell}, U_{R_\ell}, U_{R_{\ell-1}}, \dots, U_{R_1} \} $ satisfies}:
    \begin{enumerate}
        \item \textcolor{black}{$U =U_{L_\ell} U_{R_\ell}U_{R_{\ell-1}}\dots U_{R_1}$.} 
        \item For any $ U_{R_i}, 1\leq i \leq \ell $, it has non-zero diagonals only at diagonals with indices $ \{0,\pm \lceil r_c^{(i)} \rceil\} $ where $ r_c^{(i)}=a\cdot \lceil \frac{r/a}{2^i}\rceil $.
        
        
        \item The non-zero diagonals of $U_{L_\ell} $ are distributed between $[-r', r']$, where $ r'=r-\sum_{i=1}^{\ell} r_c^{(i)} $.
    \end{enumerate}
    If $f$ is an ideal depth-$\lfloor\log{(r/a)}\rfloor$ decomposition, then we may also call it an ideal full-depth decomposition.
    
\end{definition}

However, intuitively, not all permutations possess such an ideal decomposition, and it is even non-trivial to check whether a permutation has one. To our knowledge, the most relevant work is that of Han \textit{et al.} \cite{han2019discrete}, which only proposes a decomposition of the discrete Fourier transformation (DFT) satisfying Definition \ref{def: idealdmp}   rather than any pure permutation. 




\section{Searching Ideal Depth-1 Decomposition} \label{sec: singleround}

\textcolor{black}{To find a full-depth ideal decomposition for a given permutation, we can first construct a method searching for a depth-$1$ solution.  This initial search can then be recursively applied to the resulting factors to achieve further decomposition. Moreover, the depth‑$1$ decomposition factors themselves may reveal interpretable structural patterns that serve as evidence of full-depth decomposability.
This section aims to complete the first step, designing a depth-$1$ ideal decomposition search for any permutation. 
}


\subsection{Basic Decomposition Pattern}

To factorize an $n\times n$ permutation matrix $U$ into a product of two factors ($U=U_L U_R$) satisfying the ideal depth-$1$ decomposition of Definition \ref{def: idealdmp},
\textcolor{black}{we first observe how non-zero entries on diagonals of $U_L$ and $U_R$ contribute to those in $U$. For any non-zero entry $U(k,l)^*$ in $U$, it is determined by the $k_R$-th diagonal in $U_R$ and the $k_L$-th diagonal in $U_L$, where $k_L = k - k_R \bmod n$ (Theorem \ref{thm: k=k1+k2}). This follows that $U_R(k_R,k+l-k_R)^*$ and $U_L(k_L,l)^*$ should be set to $1$ to guarantee the correct generation of $U(k,l)^*$. We call this decomposing $U(k,l)$ into $k_R$-th diagonal and $k_L$-th diagonal (see Figure \ref{fig: ConvertR}).} Though entry $U(k,l)$ can be associated with more than two diagonals in general, it will end up with $U_R$ and $U_L$ not being permutations, which may damage the composability.

In the following context, we use function $\texttt{ConvertR}$ to describe the operation of decomposing any entry $ U(k,l)^* $ into the $ k_R $-th diagonal of $ U_R $:
\begin{equation}
    \texttt{ConvertR}(k,l,k_R) = k+l-k_R.
\end{equation}
The output is the row index of the entry in $U_R$ that will be set to $1$. 

During decomposing entries of $U$ into some desired diagonals of $U_R$ and $U_L$, conflicts may occur and damage the correctness of factorization. 
\textcolor{black}{Specifically, we say there is a row (or column) conflict in the matrix if there is more than one non-zero entry in a row (or column).}  
In $U_R$, row conflicts may occur, while column conflicts are impossible; conversely, the opposite holds in $ U_L $. Additionally, since the column index of the entry to be set to $ 1 $ in $ U_L $ equals \textcolor{black}{$\texttt{ConvertR}(k, l, k_R)$}, a duality exists between the two matrices, indicating that managing either $ U_L $ or $ U_R $ is sufficient to make the other matrix undergo a mirroring process.

\begin{figure}[h]
    \centering
    \includegraphics[scale=0.44]{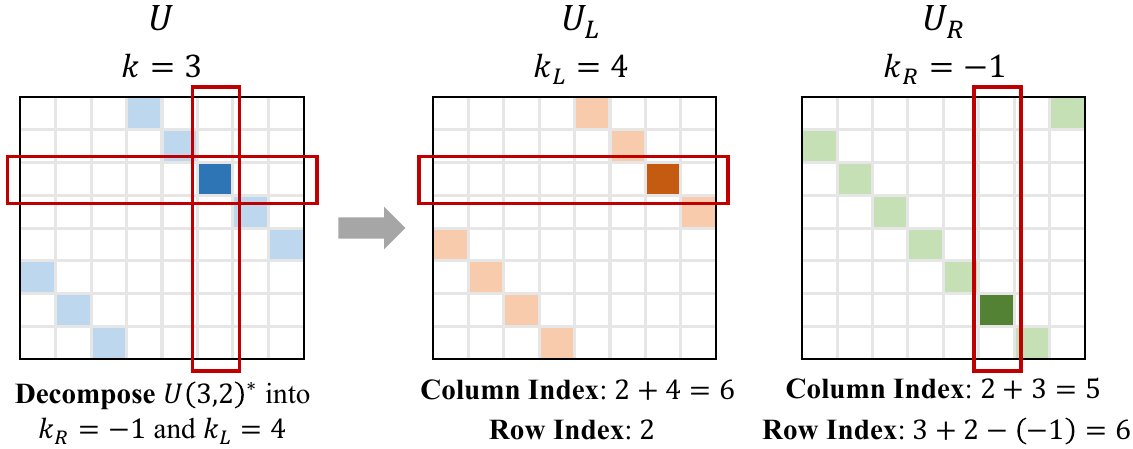}
    \caption{Decompose single entry in $8\times 8$ permutation matrix}
    \label{fig: ConvertR}
    \Description[Decomposing single entry in an $8\times 8$ permutation matrix]{} 
\end{figure}

\begin{algorithm}[t]
    \renewcommand{\algorithmicrequire}{\textbf{Input:}}
    \renewcommand{\algorithmicensure}{\textbf{Output:}}
    \caption{\textcolor{black}{Depth-$1$ Ideal Decomposition Search}} 
    \label{alg:IdealDmpSearch}
    \begin{algorithmic}[1]
    \Require 
        \Statex $U$: Permutation Matrix of size $n\times n$; 
        \Statex $a$: Common difference of the non-zero diagonal indices in $U$;
        \Statex $r$: Distribution bound of the non-zero diagonal indices in $U$.
    \Ensure 
        \Statex $U_L,U_R$: Depth-$1$ Ideal decomposition factors of $U$.  
    \vspace{0.3em} 
    \Statex \textbf{[Step 1: Decompose entries]}
    \State $r_c \leftarrow a\lceil \frac{r/a}{2} \rceil$
    \State $M_{r_c},M_{-r_c},M_{0}\leftarrow \varnothing$ 
    \Comment{Initialize empty tables}
    \State $M\leftarrow \{M_{r_c},M_{-r_c},M_{0}\}$

    \For{$k$ traversing indices of non-zero diagonals in $U$}
        \For{$l$ traversing row indices of non-zero entries in the $k$-th diagonal}
            \If{$M_{-r_c}[\texttt{ConvertR}(k,l,r_c)]=\varnothing$ and $k>=r_c$}
                \State $M_{r_c}[\texttt{ConvertR}(k,l,r_c)]\leftarrow k$
            \ElsIf{$M_{r_c}[\texttt{ConvertR}(k,l,-r_c)]=\varnothing$ and $k<=r_c$} 
                \State $M_{-r_c}[\texttt{ConvertR}(k,l,-r_c)]\leftarrow k$ 
            \ElsIf{$-r_c<k<r_c$} 
                \State $M_{0}[\texttt{ConvertR}(k,l,0)\leftarrow k$
            \Else
                \State \textbf{return} $\varnothing,\varnothing$ 
                \Comment{Decomposition not exists}
            \EndIf
        \EndFor
    \EndFor
    
    \vspace{0.3em} 
    \Statex \textbf{[Step 2: Resolve conflicts]}
    \State $Q\leftarrow \varnothing$ 
    \Comment{Initialize empty list}
    \For{ $a$ traversing the keys in $M_0$} 
        \If{$M_{-r_c}[a]\not=\varnothing$ or $M_{r_c}[a]\not=\varnothing$}
            \State Append $a$ to $Q$ 
            \Comment{Add one conflict to list}
        \EndIf
    \EndFor
    \State $solved \leftarrow 1$ 
    \If{$|Q| \not= 0 $}
        \State $solved \leftarrow \texttt{CheckRow}(M,n,r_c,Q[0],Q,0)$
    \EndIf

    \vspace{0.3em}
    \Statex \textbf{[Step 3: Reconstruct $U_L$]}
    \If{$solved = 0$}
        \State \textbf{return} $\varnothing,\varnothing$
        \Comment{Decomposition not exists}
    \EndIf
    \State Construct $U_R$ from $M$, Compute $U_L$ such that $U=U_LU_R$ 
    \State \textbf{return} $U_L,U_R$
    \Comment{Decomposition exists}
    \end{algorithmic}
\end{algorithm}

\begin{algorithm}[t]
    \renewcommand{\algorithmicrequire}{\textbf{Input:}}
    \renewcommand{\algorithmicensure}{\textbf{Output:}}
    \caption{Check One Conflicting Row Between the $0$-th and $\pm r_c$-th Diagonals (\texttt{CheckRow})}
    \label{alg:CheckRow}
    \begin{algorithmic}[1]
    \Require 
        \Statex $M = \{ M_{r_c}, M_{-r_c}, M_{0} \}$: Non-zero diagonals in $U_R$;
        \Statex $n$: The dimension of $U$;
        \Statex $r_c$: The non-zero diagonal index of $U_R$;
        \Statex $a_0$: The row with a potential conflict to be resolved;
        \Statex $Q$: The set containing indices of rows where conflicts exist;
        \Statex $i$: The number of resolved rows in $Q$.
    \Ensure 
        \Statex $solved \in \{0,1\}$ 

    \If{$M_0[a] = \varnothing$}
        \Comment{Row $a$ has no conflict}
        \If{$i + 1 < \textproc{len}(Q)$}
            \Comment{Find next conflict to resolve}
            \State $i\leftarrow i+1, a\leftarrow Q[i]$         
        \Else
            \Comment{All conflicts resolved}
            \State \textbf{return} $1$ 
        \EndIf
    \EndIf
    
    \vspace{0.3em} 
    \Statex \textbf{[Try to move $U_R(0,a)^*$ to the $r_c$-th diagonal]}
    \State $b\leftarrow M_0[a],M_0[a]\leftarrow \varnothing, \hat{a}\leftarrow (a-r_c)\bmod{n}, solved\leftarrow0$ 
    \If{$b \geq 0$ and $M_{-r_c}[\hat{a}] = \varnothing$}
        \State $M_{r_c}[\hat{a}]\leftarrow b$
        \State $solved\leftarrow \texttt{CheckRow} (M,n,r_c,\hat{a},Q,i)$ 
        \Comment{Resolve new conflicts caused by the movement}
    \EndIf
    
    \vspace{0.3em} 
    \Statex \textbf{[Otherwise, try to move $U_R(0,a)^*$ to the $-r_c$-th diagonal]}
    \If{$solved = 0$ and $b \leq 0$}
        \State $M_{r_c}[\hat{a}]\leftarrow \varnothing, \hat{a}\leftarrow a+r_c\bmod{n}$
        \If{$M_{r_c}[\hat{a}] = \varnothing$}
            \State $M_{-r_c}[\hat{a}]\leftarrow b$
            \State $solved\leftarrow \texttt{CheckRow}(M,n,r_c,\hat{a},Q,i)$ 
            \Comment{Resolve new conflicts caused by the movement}
        \EndIf
    \EndIf
    \If{$solved = 0$} 
    \Comment{Cannot resolve, undo modifications}
        \State $M_{r_c}[\hat{a}] \leftarrow \varnothing, M_{-r_c}[\hat{a}] \leftarrow \varnothing, M_{0}[a] \leftarrow b$
    \EndIf
    
    \State \textbf{return} $solved$ 
    \Comment{Return to the previous recursive call}
    \end{algorithmic}
\end{algorithm}

\subsection{Concrete Design}

Based on the observation above, we design an algorithm searching a depth-1 ideal decomposition $U=U_LU_R$ for \textcolor{black}{any} $n\times n$ permutation $U$. The algorithm consists of three steps, as outlined below. \textcolor{black}{Its pseudocode is provided in Algorithm~\ref{alg:IdealDmpSearch}.}  

\noindent\textbf{Step 1. (Decompose Entries)} This step decomposes each non-zero entry in $U$ to the $ r_c $-th, $ -r_c $-th, or $0$-th diagonals of $ U_R $. \textcolor{black}{$r_c$ corresponds to $r_c^{(1)}=a\cdot \lceil \frac{r/a}{2}\rceil $ in Definition \ref{def: idealdmp}, where $U$'s non-zero diagonal indices lie in $[-r,r]$ and $a$ is the greatest common difference found.} First, we create a set of key-value tables representing the non-zero diagonals of $U_R$: 
$$M= \{ M_{r_c}, M_{-r_c}, M_{0} \}, $$ where $M_i$ represents the $i$-th non-zero diagonal for $i\in \{\pm r_c,0 \}$. For each table, its keys indicate the row coordinates of entries on it, and its values denote the diagonal coordinates of entries in $U$ determined by the keys. Next, for each non-zero entry $U(k,l)^*$ in $U$:
\begin{enumerate}[leftmargin=1.7em]
    \item If \( k \geq r_c \) and \( M_{-r_c}[\texttt{ConvertR}(k, l, r_c)] \) is empty, add the key-value pair \( \langle \texttt{ConvertR}(k, l, r_c), k \rangle \) to \( M_{r_c} \). 
    
    \item If \( k \leq -r_c \) and \( M_{r_c}[\texttt{ConvertR}(k, l, -r_c)] \) is empty, add the key-value pair \( \langle \texttt{ConvertR}(k, l, -r_c), k \rangle \) to \( M_{-r_c} \).  
    
    \item If \( -r_c < k < r_c \), add \( \langle \texttt{ConvertR}(k, l, 0), k \rangle \) to \( M_0 \);
    
    \item Otherwise, a conflict that cannot be resolved exists, and the algorithm terminates. 
\end{enumerate}
The explanation of the conditional statements above is given as follows:

According to Theorem \ref{thm: k=k1+k2} and Definition \ref{def: idealdmp}, for the diagonal distribution of $ U_L $ to converge to $ [-r + r_c, r - r_c] $ with $ r_c = a \cdot \lceil \frac{r/a}{2} \rceil $, non-zero entries on the $k$-th diagonal of $ U $ with $ |k|> r_c $ can only be decomposed to the $\texttt{sign}(k)\cdot r_c $-th diagonal of $ U_R $. Step 1 employs the first two conditional statements to receive these entries for decomposition. If conflicts arise between the $ r_c $-th and $ -r_c $-th diagonals during the decomposition of these entries, it indicates that no ideal decomposition exists for $ U $. Such cases are captured by the fourth conditional statement. 

The remaining non-zero entries of $ U $ are distributed on diagonals whose indices' absolute value do not exceed $ r_c $, and these entries can either be decomposed to the 0-th diagonal or the nearest diagonal among the $ r_c $-th and $ -r_c $-th diagonals according to Theorem \ref{thm: k=k1+k2}. The strategy in Step 1 is to decompose all of them to the 0-th diagonal. 

After Step 1, $ U_R $ has non-zero entries only on the 0-th and $ \pm r_c $-th diagonals, with potential conflicts occurring only between the 0-th and $ r_c $-th diagonals, as well as between the 0-th and $ -r_c $-th diagonals. 


\noindent\textbf{Step 2. (Resolve Conflicts)} This step aims to resolve the conflicts derived from the previous step. We first create a set $ Q $ to record the rows where conflicts exist. For each key-value pair $ \langle a_1, b_1 \rangle $ in $ M_0 $, if there exists another key-value pair $ \langle a_1, b_2 \rangle $ in $ M_{r_c} $ or $ M_{-r_c} $, it indicates a conflict in row $ a_1 $ of $ U_R $, and we add $ a_1 $ to the set $ Q $. Then, we compute $ solved \leftarrow \texttt{CheckRow}(M, n, r_c,Q[0], Q, 0) $ (Algorithm \ref{alg:CheckRow}). If $ solved = 1 $, all the conflicts are successfully resolved; otherwise, there are irresolvable conflicts, and the algorithm terminates.

Assigned with a specific row index $ a $, the function $ \texttt{CheckRow} $ first checks if any conflict exists in row $ a $. If there is a conflict, it attempts to resolve the conflict; otherwise, it proceeds to resolve the next conflict in $ Q $. From the previous step, it is known that the conflicts occur between the 0-th diagonal and the $ \pm r_c $-th diagonals. Therefore, for any conflicting row $ a $, $ \texttt{CheckRow} $ can only resolve it by moving the entry $U_R(0,a)^*$ to the $ r_c $-th (or $ -r_c $-th) diagonal. The condition for successfully moving $U_R(0,a)^*$ is that the target row $ \hat{a} $ must not contain any entries on the $ r_c $-th (or $ -r_c $-th) diagonal. If such an entry exists (with a corresponding key-value pair $ \langle \hat{a}, v \rangle $ in $ M_{k} $ for $ k = r_c $ or $ -r_c $), two possible cases arise: 
\begin{enumerate}[leftmargin=1.7em]
    \item If $ |v| > r_c $, then $U_R(0,a)^*$ conflicts with $U_R(k,\hat{a})^*$ after the movement, but $U_R(k,\hat{a})^*$ is not movable. 
    \item Otherwise, $U_R(k,\hat{a})^*$ must have originally been moved from the $ 0 $-th diagonal to row $\hat{a}$. If continuing to move this entry results in an ideal decomposition, \textcolor{black}{then there must be an alternative order of movements for entries on the $0$-th diagonal that does not move any entry more than once and reach the same solution.} 
\end{enumerate}
Therefore, $U_R(0,a)^*$ should not be moved to row $\hat{a}$ in either of these cases. If the movement is successful, $\texttt{CheckRow}$ invokes itself to resolve any potential conflict between the $r_c$-th (or $-r_c$-th) diagonal and the $0$-th diagonal for row $\hat{a}$. If the invocation returns $0$ for irresolvable conflicts, then $\texttt{CheckRow}$ has to undo the movement it employs and return $0$ to the previous recursive call. If the invocation determines no conflict exists in row $\hat{a}$, it proceeds to resolve the next conflict recorded in $Q$. In this way, $\texttt{CheckRow}$ operates as a depth-first search, exhaustively checking all entry movement combinations until it identifies the first conflict-free $U_R$. 

\noindent\textbf{Step 3. (Reconstruct $U_L$)}. Generate $U_L$ according to $U$ and $U_R$ such that $U = U_L U_R$. One may replace $U$ with $U_L$ and update $r_c$ for depth-$2$ ideal decomposition search. 

The correctness of the search algorithm in detecting whether a given permutation $U$ admits any ideal depth-1 decomposition is demonstrated throughout the description above. In the next section, the search algorithm is applied to specific permutations in homomorphic matrix operations, where we provide a visualizable demonstration of this algorithm.

\textcolor{black}{Note that a given permutation may admit none, one, or multiple ideal depth-1 decompositions, and the same holds for each of its possible resulting factors $U_L$. Thus, a straightforward recursive invocation of our search algorithm does not guarantee an ideal decomposition of maximal depth. However, the search algorithm is capable of finding all ideal depth-1 solutions for $U$. This is achieved by slightly adapting the algorithm to treat the final resolved conflict as unresolvable every time, prompting it to backtrack through previous conflicts for alternative solutions. By further applying the search algorithm to the $U_L$s of all identified solutions, we can determine whether $U$ admits a depth-$2$ decomposition. This approach naturally extends to a depth-first search strategy for obtaining $U$'s maximal ideal decomposition depth.}

\section{Full-depth Ideal Decomposability on Specific Permutations}\label{sec: DMP4HMTHMM}
In this section, we demonstrate that certain permutations can be ideally decomposed in full depth as per Definition \ref{def: idealdmp}. 
We first apply the depth-$1$ decomposition searching algorithm from Section \ref{sec: singleround} to the permutation used in Jiang \textit{et al.}'s HMT \cite{jiang2018secure}, through which we discover the full-depth ideal decomposability of this permutation. We further reduce the permutations from Rizomiliotis and Triakosia's HMM \cite{rizomiliotis2022matrix} to the former, illustrating their decomposability. 


\subsection{Homomorphic Matrix Transposition}\label{sec: HMTDmp}

Let $ ct(\mathbf{a}) $ denote a ciphertext of an $ n $-dimensional vector $\mathbf{a}$, where $\mathbf{a}$ is the row-ordering encoding of a $ d \times d $ matrix $ A $ \textcolor{black}{($d^2\leq n$)}. To obtain the ciphertext of the row-ordering vector of $ A^T $, a specific homomorphic permutation (denoted as \textcolor{black}{$n \times n$ matrix} $ U^t $) is applied to $ ct(\mathbf{a}) $.  \textcolor{black}{$ U^t $ possesses $ 2d-1 $ nonzero diagonals defined as:
\begin{equation*}
\begin{aligned}
    U^t\Big(i(d-1), j(d+1)+ik\Big)^*=
    \begin{cases}
        1,&i \in [0,d-1], k = 1;\\
        1,&i \in [-d+1,0), k = -d;\\
        0,&\text{otherwise}.\\
    \end{cases}
\end{aligned}
\end{equation*}
where $j \in [0, d-|i|)$.
}
Currently, the BSGS-based HLT mentioned previously yields the best performance for $U^t$ with an overhead of approximately $ d_1 + 2d_2 $ rotations and keys, where $ d_1d_2 = d $. \textcolor{black}{As noted previously, using the Benes network for decomposition would yield a rotation complexity no less than $ O(4 \log{d^2}) $, which is inefficient within the range $ 0 < d < 2^8 $ for practical HE parameter settings \cite{bossuat2024security}, not to mention the required $ 2 (\log{d^2}-1) $ multiplicative depth.}




\subsubsection{Case with Dimension a Power of $2$}\label{sec: U^t_pow2}

For $d$ a power of $2$, $U^t$ can be precisely decomposed to $U^t_L U^t_R$ when we apply the depth-$1$ decomposition searching algorithm proposed in Section \ref{sec: singleround}. \textcolor{black}{This procedure is illustrated in Figure~\ref{fig: depth-1Dmp}, which depicts the decomposition search of a $U^t$ with size $16\times 16$ ($d=4$).} 

The $r_c$ is set to \textcolor{black}{$(d-1)\cdot \lceil \frac{(d-1)^2/(d-1)}{2} \rceil = 6$} as the indices of non-zero diagonals in $U^t$ have a common difference of $d-1=3$.  Step 1 of the algorithm is presented in Figure \ref{fig: depth-1Dmp_2}. We can observe that $U^t$'s entries in diagonals $k\geq 6$ and $k\leq -6$ are decomposed to the $6$-th and $-6$-th diagonal of $U_R^t$ respectively, while those in diagonal $-6<k<6$ are decomposed to the $0$-th diagonal. This introduces conflicts in rows $6$ and $9$. Figure \ref{fig: depth-1Dmp_3} shows how Step 2 resolves the conflicts, illustrated by the dash paths. For example, let $\texttt{CheckRow}$ begin with row $9$: $\texttt{CheckRow}$ first moves the conflicting entry $U_R(0,9)$ to the $6$-th diagonal as $M_0[9]>0$, and this corresponds to a destination of row $3$. $\texttt{CheckRow}$ then invokes itself to check row $3$. Determining that no conflict exists at row $3$, the invocation proceeds to resolve the conflict at row $6$.

Once we solve all the conflicts in $U^t_R$, we recover $U^t_L$ such that $U^t = U^t_L U^t_R$ (Figure \ref{fig: depth-1Dmp_4}). We can observe that both of these factors exhibit good structural properties.
\begin{figure*}[h]
    \centering
        \subfloat[$16\times 16$ $U^t$ with non-zero entries painted in dark blue.\label{fig: depth-1Dmp_1}]{\includegraphics[width=0.5\columnwidth]{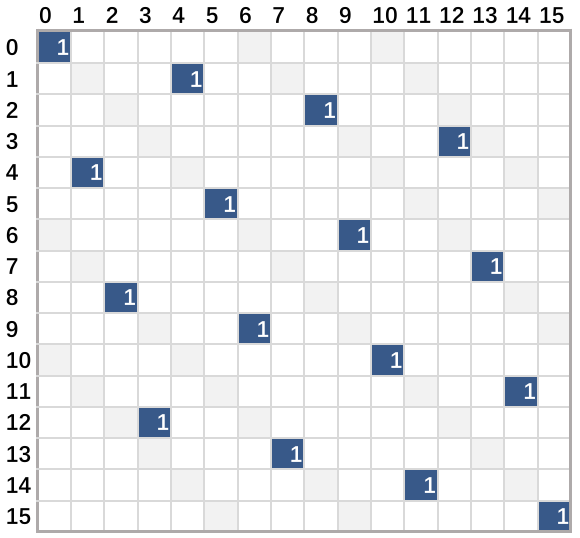}}
        \hfil 
        \subfloat[Step 1. Red entries denote the non-zero entries in $U^t_R$ after Step 1.\label{fig: depth-1Dmp_2}]{\includegraphics[width=0.5\columnwidth]{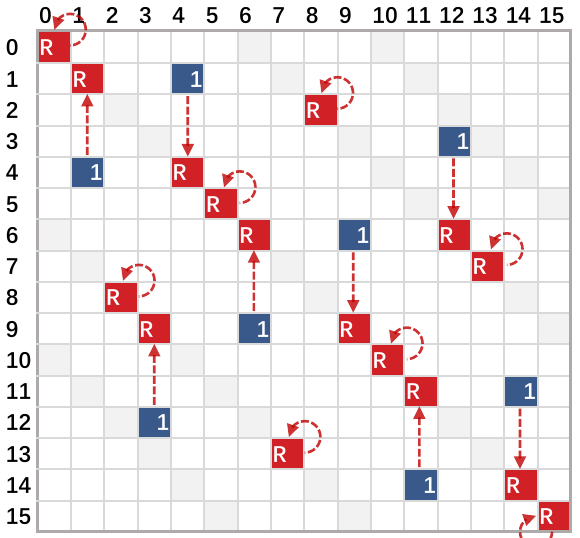}}
        \hfil
        \subfloat[Step 2. $U_R^t$'s final view. Conflicting entries moved away are painted in grey\label{fig: depth-1Dmp_3}]{\includegraphics[width=0.5\columnwidth]{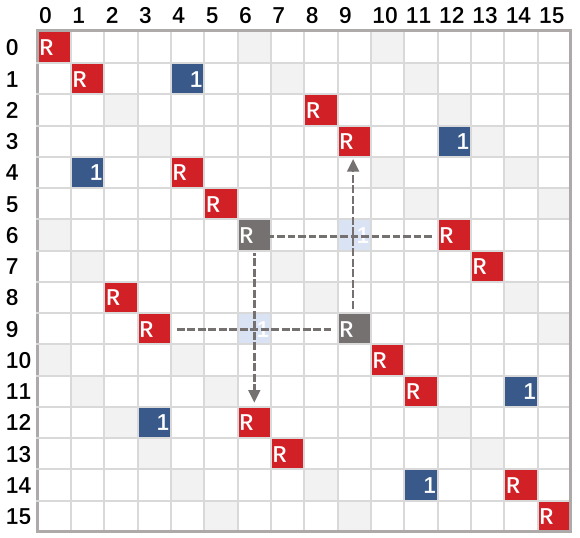}}
        \hfil 
        \subfloat[Step 3. $U^t_L$ reconstructed by $U^\sigma_R$ and $U^t$.\label{fig: depth-1Dmp_4}]{\includegraphics[width=0.5\columnwidth]{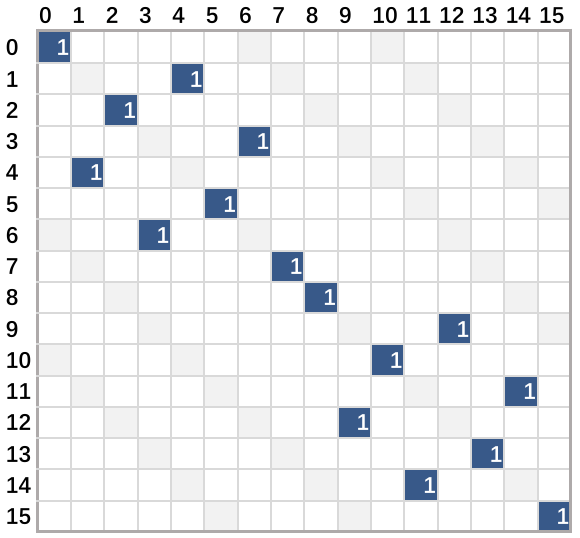}}
    \caption{Applying depth-1 ideal decomposition search to $U^t$ of size $16\times 16$.}
    \label{fig: depth-1Dmp}
\end{figure*}
If we let $A$ be evenly partitioned into $4$ blocks with size $d/2\times d/2$:
\begin{equation} \label{eq: A_org}
    A = 
    \left[
    \begin{matrix}
        {A'_{0,0}} & {A'_{0,1}} \\
        {A'_{1,0}} & {A'_{1,1}} \\
    \end{matrix}
    \right],    
\end{equation}
then the effect of $U^t_R$ and $U^t_L$ can be expressed as follows:
\begin{equation} \label{eq: U^tSingleDmp}
\left[
\begin{matrix}
    {A'_{0,0}} & {A'_{1,0}} \\
    {A'_{0,1}} & {A'_{1,1}} \\
\end{matrix}
\right]
\leftarrow U^t_R(A), 
\left[
\begin{matrix}
    {A'_{0,0}}^T & {A'_{1,0}}^T \\
    {A'_{0,1}}^T & {A'_{1,1}}^T \\
\end{matrix}
\right]
\leftarrow U^t_L\left(U^t_R(A)\right).
\end{equation}
Equation \ref{eq: U^tSingleDmp} indicates that $U^t_R$ swaps the positions of $A'_{0,1}$ and $A'_{1,0}$, while $U^t_L$ performs block-wise transposition. \textcolor{black}{It follows that the equation}  
\begin{equation}\label{eq: U^tMultiDmp}
A^T = U_{L_\ell}^t\left( U_{R_\ell}^t\left(\dots \left(U_{R_1}^t(A) \right) \right) \right)    
\end{equation}
\textcolor{black}{holds for $\ell \leq \lfloor \log{(d-1)} \rfloor$ if the process below defines the permutations in the equation:}
\begin{enumerate}[leftmargin=1.7em]
    \item \textcolor{black}{$U^t_{R_i}$ evenly partitions each block of $A$ into $4$ sub-blocks (with $A$ itself being the initial block). The sub-blocks are interpreted as blocks in the next partition before $U^t_{R_{i+1}}$.}\label{enu: rule1}
    \item $U^t_{R_i}$ swaps the sub-blocks at positions $[0,1]$ and $[1,0]$ in each block of $A$.
    \item $U^t_{L_\ell}$ performs block-wise transposition.
\end{enumerate}

\textcolor{black}{These rules indicate that $U^t_{R_i}$ swaps $d/2^i$-size sub-blocks inside each $d/2^{i-1}$-size block, where its non-zero diagonals are precisely distributed at $\{0, \pm (d-1)d/2^{i}\}$ according to Theorem \ref{thm: 2BS} and Theorem \ref{thm: BS4S}. On the other hand, $U^t_{L_\ell}$ transposes each $d/2^\ell$-size block by swapping their entries at positions $[q,p]$ and $[p,q]$ for $0 \leq p < q < d/2^{\ell}$. These swappings correspond to non-zero diagonals with indices:}
\begin{equation*}
    \begin{split}
         &\{\pm(d(q-p) + p-q) \mid 0 \leq p < q < d/2^{\ell}) \}=\\
         &\{\pm i \cdot (d-1) \mid 0 \leq i < d/2^{\ell}\}.
    \end{split}
\end{equation*}
Now, it is evident that Equation \ref{eq: U^tMultiDmp} with $U^t_{R_i}$ and $U^t_{\ell}$ defined above 
is a depth-$\ell$ ideal decomposition of $U^t$ for $1\leq \ell \leq \lfloor\log{(d-1)}\rfloor$. 
\textcolor{black}{The concrete rotation complexity for evaluating a depth-$\ell$ decomposed $U^t$ is $2\ell + 2d_2 +d_1$. Specifically, the factor $U_{L_\ell}$ has around $2d/2^\ell$ non-zero diagonals symmetrically distributed on both sides of the $0$-th diagonal, and is computed using HLT with $d_1$ BSGS inner loops and $2d_2$ outer loops. Other factors, $U_{R_1},...,U_{R_\ell}$, are each computed via HLT without the BSGS algorithm, contributing $2\ell$ rotations.} 



\begin{theorem}\label{thm: 2BS} 
The permutation matrix swapping any two blocks 
of any 
matrix contains at most three nonzero diagonals.
\end{theorem}
\begin{proof}
Given $ A $ is of size \textcolor{black}{$c \times d$}, and the two same size blocks to be swapped are $ A'_{i,j}$ and $A'_{k,l}$. If $ A'_{i,j}[0,0] = A[p,q] $ and $ A'_{k,l}[0,0] = A[u,v] $, for any two entries $A'_{i,j}[x,y]$ and $A'_{k,l}[x,y]$, we have:
\begin{equation}
    \begin{split}
        A'_{i,j}[x,y] = A[p+x, q+y], \\
        A'_{k,l}[x,y] = A[u+x, v+y].      
    \end{split}
\end{equation}
Let $ r = (u+x)d + (v+y) - (p+x)d - (q+y) = d(u-p) + (v-q) $. Moving $A'_{i,j}$ to the position of $A'_{k,l}$ requires a rotation by step $-r$, while moving $A'_{k,l}$ to the position of $A'_{i,j}$ requires rotation by step $r$. Entries outside the blocks maintain their original positions requiring rotation by step $0$. Furthermore, if either $A'_{i,j}$ or $A'_{k,l}$ is entirely zero, the corresponding rotation for that block can be omitted. Thus, the swapping involves at most three non-zero diagonals in its matrix representation.
\end{proof}

\begin{theorem}\label{thm: BS4S} 
The permutations that perform the same operation on different blocks of the same size within a matrix have identical non-zero diagonal distributions.
\end{theorem}
\begin{proof}
Given $ A $ is of size $c \times d$, let $ A' $ be any submatrix of $ A $ with $A'[0,0]=A[s,t]$. For any two entries $ A'[p,q] $ and $ A'[u,v] $, we have:
\begin{equation}
    \begin{split}
        A'[p,q] = A[s+p,t+q], \\
        A'[u,v] = A[s+u,t+v].
    \end{split}
\end{equation}
To move $ A'[p,q] $ to the position of $ A'[u,v] $, it requires a rotation of $ d(s+u)+(t+v) - d(s+p)-(t+q)= d(p-u)+(q-v) $ steps. It is evident that the number of different rotations for movements within the block is independent of the block's position within $ A $. 
\end{proof}

\subsubsection{Case with Arbitrary Dimension} \label{sec: U^t_arb}
We observe that when the search algorithm is recursively applied to $U^t$ with arbitrary dimension $d$, it always terminates after $k$ recursion, where $k$ is the largest power of $2$ dividing $d$. This aligns with the decomposition rule for power-of-2 cases we proved above: $U^t_{R_{k+1}}$ is not defined as the blocks in $A$ cannot be evenly partitioned into $4$ sub-blocks at that point. 


The simplest solution is to pad \( A \) with zeros, increasing its dimensions to the nearest power of 2, \( 2^i \). 
\textcolor{black}{Then the corresponding $U^t$ can undergo an ideal decomposition of depth $\lceil \log{d} \rceil-1$.} 
\textcolor{black}{This approach is practical unless $ (2^{i})^2$ exceeds the maximum encoding length \( n \) of the ciphertext.} 
\textcolor{black}{We propose an alternative method to avoid this drawback.} 
\textcolor{black}{The idea is to define a new method for $U^t_{R_i}$ to partition and swap odd-dimensional blocks.} 
Specifically, for any $d\times d$ block $A$ with $ d = 2 \cdot d' + 1 $ for some $ 0 \leq d' $, let:
\begin{equation}\label{eq: A_divide0}
\begin{array}{ll}
B'_{0,0} = 
\begin{bmatrix}
A[i,j]
\end{bmatrix}_{0 \leq i \leq d'}^{0 \leq j \leq d'}
&
B'_{0,1} =
\begin{bmatrix}
A[i,j]
\end{bmatrix}_{0 \leq i \leq d'-1}^{d'+1 \leq j \leq 2d'}
\\
\\
B'_{1,0} =
\begin{bmatrix}
A[i,j]
\end{bmatrix}_{d'+1 \leq i \leq 2d'}^{0 \leq j \leq d'-1}
&
B'_{1,1} =
\begin{bmatrix}
A[i,j]
\end{bmatrix}_{d' \leq i \leq 2d'}^{d' \leq j \leq 2d'}
\end{array}.
\end{equation}
The partition of $A$ into unequal blocks is given as:
\begin{equation}\label{eq: A_divide}
A=
\begin{bmatrix}
B'_{0,0} & B'_{0,1} \\
B'_{1,0} & B'_{1,1}
\end{bmatrix}
,
\end{equation}
\textcolor{black}{with an overlapping part} $A[d',d']$ between $B'_{0,0}$ and $B'_{1,1}$. 
In any partitioning round, \textcolor{black}{the blocks have at most} two dimensions differing by at most 1 (proof given in Appendix \ref{thm: Only2Sizes}). \textcolor{black}{In this case, the non-zero diagonals in the block-wise permutation are the union of those required for the permutations on blocks of different sizes.} \textcolor{black}{Thus, for $1 < i \leq \ell$, the maximum number of non-zero diagonals in each $U^t_{R_i}$ increases from 3 to 5. The number of non-zero diagonals of $U^t_{L_\ell}$ remains $O(d/2^{\ell-1})$ for $\ell \leq \lceil \log{d} \rceil -1$. Such a decomposition closely approximates the ideal form and avoids the possible increment of $n$ from the zero-padding strategy.} 

\subsection{Homomorphic Matrix Multiplication}\label{sec: HMMDmp}

It is intriguing that the permutations in Rizomiliotis and Triakosia's HMM algorithm \cite{rizomiliotis2022matrix} have decomposability similar to $U^t$. We can achieve an asymptotically faster HMM by applying decompositions to these permutations. Furthermore, we present methods to reduce the multiplication depth incurred in the decomposition to a constant level and to enhance the utilization of the encoding space.


Let $A$ and $B$ be the $d\times d$ matrices to be multiplied, and let the plaintext space be a vector space of dimension $n=d^3$. The method proposed by Rizomiliotis and Triakosia first appends $(d-1)$ zero matrices of size $d \times d$ to both the end of $A$ and $B$ 
to form $\hat{A}$ and $\hat{B}$ of size $d^2\times d$, where each of them can be encrypted by a ciphertext. 
Two permutations $U^\gamma$ and $U^\xi$ are then applied to $\hat{A}$ and $\hat{B}$ respectively. Next, column-wise replication is performed on $U^\gamma(\hat{A})$ and row-wise replication on $U^\xi(\hat{B})$. Denote the replicated results as $\tilde{A}$ and $\tilde{B}$. Finally, compute $\hat{C} = \tilde{A} \odot \tilde{B}$, and perform $\hat{C} \leftarrow \texttt{Rot}(\hat{C}, d^3 / 2^i) + \hat{C}$ for $1 \leq i \leq \log{d}$. This yields the desired result $C = AB$ in the first $d \times d$ block of $\hat{C}$.

Compared to previous HMM algorithms \cite{halevi2018faster,jiang2018secure}, Rizomiliotis and Triakosia's approach is characterized by its novel use of a $d^3$ redundant encoding space, which minimizes the number of ciphertext multiplications and rotations.
Its complexity is considered to be dominated by the $ U^\gamma $ and $ U^\xi $ permutations, which necessitate $ 2d $ ciphertext rotations.



\subsubsection{Decomposability of $U^\gamma$ and $U^\xi$}


We observe that $U^\gamma$ can be fully regarded as a transposition acting on a $d\times d$ matrix where each entry is a $d\times 1$ vector:
\begin{equation}
\begin{split}
& \hat{A} = \left[
        \begin{matrix}
            \mathbf{a}_{0}       & \cdots    & \mathbf{a}_{d-1}     \\
            \mathbf{0}_c         & \cdots    & \mathbf{0}_c            \\
            \vdots               &   \ddots  & \vdots       \\
            \mathbf{0}_c         & \cdots    & \mathbf{0}_c            \\
        \end{matrix}
    \right],
U^\gamma(\hat{A}) = \left[
            \begin{matrix}
                \mathbf{a}_{0}     & \mathbf{0}_c         & \cdots     & \mathbf{0}_c         \\
                \vdots     & \vdots    &     \ddots       & \vdots    \\ 
                \mathbf{a}_{d-1}   & \mathbf{0}_c         & \cdots     & \mathbf{0}_c         \\
            \end{matrix}
        \right], \\ 
& \mathbf{a}_i=
    \left[
        \begin{matrix}
            A[0,i] & \cdots & A[d-1,i]
        \end{matrix}
    \right]^T, \text{for} \phantom{a} 0\leq i <d.
\end{split}
\end{equation}
\textcolor{black}{The rule from the previous section can be applied to this transposition. Thus, the decomposition of $U^\gamma$ is similar to $U^t$.} Specifically, when $d$ is a power of two, $U^\gamma$ has a depth-$\ell$ ideal decomposition for $\ell \leq \lfloor\log{(d-1)}\rfloor$. When $d$ is set arbitrarily, $U^\gamma$ has an approximate depth-$\ell$ ideal decomposition for $\ell \leq \lfloor \log{d} \rceil - 1$. 

Furthermore, we can modify the decomposition of $U^\gamma$ to achieve constant multiplication depth. We fill the non-zero diagonals of $U^\gamma$ with $1$ and denote the result as $U^\gamma_{pad}$. At this point, $U^\gamma_{pad}(\hat{A})$ is the sum of the rotated $\hat{A}$ by steps of $0, -(d^2-1), \dots, -(d^2-1)(d-1)$ without any masking. It can be observed that  $U^\gamma_{pad}(\hat{A})$ has the same first column as $U^\gamma(\hat{A})$. This implies that we can get $U^\gamma(\hat{A})$ by masking $U^\gamma_{pad}(\hat{A})$ only once.
The decomposition of $U^\gamma_{pad}$ can be directly obtained by \textcolor{black}{filling non-zero diagonals of the factors in $U^\gamma$'s ideal decomposition with $1$, yielding a constant multiplication depth of $1$ (concrete construction provided in Appendix \ref{app: 1-pad}).} 

Similarly, $U^\xi$ can be interpreted as a transposition acting on a $d\times d$ matrix where each basic unit is a $1\times d$ vector:
\begin{equation}
\begin{split}
        & \hat{B} = 
        \left[
            \begin{matrix}
                \mathbf{b}_0      & \cdots    & \mathbf{b}_{d-1}     \\
                \mathbf{0}_r      & \cdots    & \mathbf{0}_r            \\
                \vdots            & \ddots    & \vdots       \\
                \mathbf{0}_r      & \cdots    & \mathbf{0}_r            \\
            \end{matrix}
        \right], 
        U^\xi(\hat{B}) =  
        \left[
            \begin{matrix}
                \mathbf{b}_0     & \mathbf{0}_r         & \cdots     & \mathbf{0}_r         \\
                \vdots           & \vdots               &  \ddots    & \vdots    \\ 
                \mathbf{b}_{d-1} & \mathbf{0}_r         & \cdots     & \mathbf{0}_r         \\
            \end{matrix}
        \right], \\
        & \mathbf{b}_i=
        \left[
            \begin{matrix}
                B[i,0] & \cdots & B[i,d-1]
            \end{matrix}
        \right], \text{for} \phantom{a} 0\leq i <d.
\end{split}
\end{equation}
Thus, $U^\xi$ can be decomposed following the decomposition method of $U^t$, and its non-zero diagonals can also be padded with $1$ to achieve a constant multiplication depth.


\subsubsection{Flexible Utilization of the Encoding Space} \label{sec: HMMDmp_space}

The original algorithm by Rizomiliotis and Triakosia necessitates a plaintext encoding space of $d^3$ to perform a $d \times d$ matrix multiplication, resulting in a fast-growing computational overhead as $d$ increases. Here, we propose a strategy to enable configurable encoding space. 
In essence, we split $U^\gamma(\hat{A})$ and $U^\xi(\hat{B})$ each into $d/d'$ \textcolor{black}{blocks (each size is $d' \times d'$ vectors)}, reducing the encoding space required for representing $A$ and $B$ to $d^2 \cdot d'$. 
Let $ d' $ be a divisor of $ d $, $ \hat{A} $ and $ \hat{B} $ are both interpreted as the concatenation of  $ d/d' $ blocks as follows:
\begin{equation}
\begin{aligned}
\hat{G} &= 
\left[
\begin{array}{c|ccc|c}
\cdots & \mathbf{g}_{d'k}    & \cdots & \mathbf{g}_{d'(k+1)-1} & \cdots \\
\cdots & \mathbf{0}      & \cdots & \mathbf{0}          & \cdots \\
\cdots & \vdots          & \ddots & \vdots              & \cdots \\
\cdots & \mathbf{0}      & \cdots & \mathbf{0}          & \cdots \\
\end{array}
\right]
\\
&=
\left[
\begin{matrix}
\hat{G}'_0 & \hat{G}'_1 & \cdots & \hat{G}'_{\frac{d}{d'}-1}
\end{matrix}
\right], {(0 \leq k < \frac{d}{d'})}
\end{aligned}
\end{equation}
where $\hat{G}=\hat{A},\mathbf{g}_i=\mathbf{a}_i,\mathbf{0} = \mathbf{0}_c$ (or $\hat{G}=\hat{B},\mathbf{g}_i=\mathbf{b}_i,\mathbf{0}=\mathbf{0}_r $). To complete an HMM, we apply block-wise $U^\gamma$ and $U^\xi$ to each $\hat{A}'_i$ and $\hat{B}'_i$ respectively, 
yielding: $U^{\gamma'}(\hat{A})=[U^{\gamma}(\hat{A}'_0) \ldots U^{\gamma}(\hat{A}'_{d/d'-1})]$ and $U^{\xi'}(\hat{B})=[U^{\xi}(\hat{B}'_0) \ldots U^{\xi}(\hat{B}'_{d/d'-1})]$. This costs $2\log{d'}$ rotations in total (Theorem \ref{thm: BS4S}). 
\textcolor{black}{For each $ U^{\gamma}(\hat{A}'_i) $ (or $ U^{\xi}(\hat{B}'_i) $), we use masking to extract it then replicate its first colum to the entire encoding space individually,}
yielding the output $ \tilde{A}_i $ (or $ \tilde{B}_i $). \textcolor{black}{These $2\frac{d}{d'}$ replications cost $2\frac{d}{d'}\log{d}$ rotations in total.}
Finally, we compute $ \hat{C} = \sum_{i=0}^{\frac{d}{d'}-1} \tilde{A}_i \odot \tilde{B}_i $, followed by a summation over $ \log d' $ rotations on $ \hat{C} $ to output the desired result. The total number of rotations is $ 3\log d' + 2\frac{d}{d'}\log d $. 
In addition, a ciphertext may store more matrices as the encoding space requirement decreases. For $m$ pairs of $d \times d$ matrices $A_i$ and $B_i$ ($0 \leq i < m$) stored in $\hat{\mathcal{A}}=[\hat{A}_0, \ldots, \hat{A}_{m-1}]^T$, the HMM between $\hat{\mathcal{A}}$ and $\hat{\mathcal{B}}$ yields 
$\hat{\mathcal{C}}=[\hat{C}_0, \ldots, \hat{C}_{m-1}]^T$, where each $\hat{C}_i$ contains $C_i = A_i B_i$. Such a parallelizable setup gives an amortized rotation complexity of $\left(\frac{3}{m}\log d' + 2\frac{d}{d'm}\log d\right)$. Further trade-offs are discussed in Appendix \ref{app: replication} for reducing the complexity component $(\frac{d}{d'}\log{d})$ of replications to $\left(\frac{d}{d_0}+\frac{d}{d'}\log{d_0}\right)$ for $1\leq d_0\leq d$.

\begin{table}
\ra{0.8}
\setlength{\tabcolsep}{2.0pt}
\begin{threeparttable}
\begin{tabular*}{\columnwidth}{@{}llcccccccccc@{}}
\toprule
Scheme & \texttt{Rot} & \texttt{Mult} & Depth & Space & Parallelity\tnote{\textdagger} \\ 
\midrule
\cite{jiang2018secure}  & $3d+5\sqrt{d}$ & $d$ & $3$ & $d^2$ & $\frac{n}{d^2}$ \\
\cite{rizomiliotis2022matrix} & $2d+3\log d$ & $1$ & $2$ & $d^3$ & -\\
\cite{huang2023secure} & $d\log{d}+d$ & $d$ & $2$ & $d^2$ & -\\
\cite{zhu2023secure} & $4d$ & $2d$ & $2$ & $d^2$ & -\\
Ours & \textcolor{black}{$3 \log{d'} + 2\frac{d}{d'}\log{d}$} & $\frac{d}{d'}$ & 2 & $d^2d'$ & $\frac{n}{d^2d'}$ \\ 
\bottomrule
\end{tabular*}
\begin{tablenotes}
    \item[\textdagger] Parallelity is not explicitly discussed in some schemes.  
\end{tablenotes}
\end{threeparttable}
\caption{Complexity of HMMs under row/column-major encoding}\label{tb: HMMcomplexity}
\end{table}

\begin{figure*}
    \centering
    \includegraphics[width=1.0\textwidth]{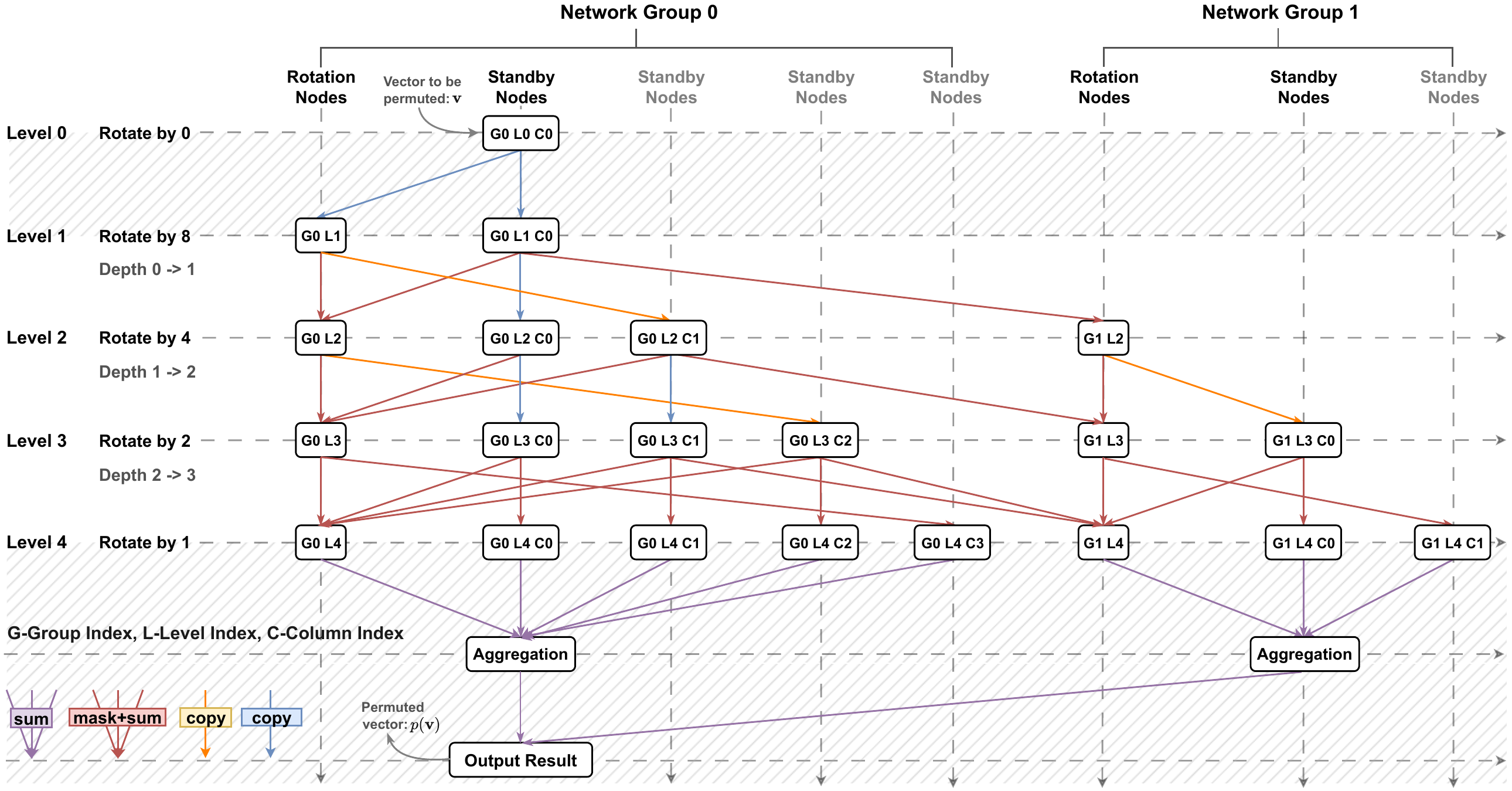}
    \caption{\textcolor{black}{A multi-group network instance constructed for a length-$16$ permutation}}
    \label{fig: network}
\end{figure*}

\subsubsection{Complexity Comparison}
Table \ref{tb: HMMcomplexity} presents existing HMMs under row (or column)-major encoding, where matrices are mapped to the plaintext (vector) space of the batch-encoding HE schemes in a row-wise or column-wise unfolded manner. 
Jiang \textit{et al.}'s scheme \cite{jiang2018secure} achieves optimal complexity when considering amortized runtime. Rizomiliotis and Triakosia \cite{rizomiliotis2022matrix} achieve a reduced number of $\texttt{Mult}$s and $\texttt{Rot}$s with a $d$-fold encoding space. However, the increased encoding space may lead to expanded ciphertexts with higher latency per basic operation. Huang \textit{et al.} \cite{huang2023secure} and Zhu \textit{et al.} \cite{zhu2023secure} design HMMs for rectangular matrices but with modest performance for square matrices. They also require a hypercube packing method supported by BGV-typed HE schemes \cite{gentry2012fully,brakerski2014leveled}, lacking compatibility with the one-dimensional encoding space. Our method does not require a hypercube structure, so it is compatible to all batch-encoding HE schemes. We also offer configurability in terms of complexity, encoding space, and rotation keys. Notably, we only need $5\log{n}$ $\texttt{Rot}$s for the maximum redundant encoding space: \textcolor{black}{$(d'=d)$. Further incorporating the proposed replication optimization yields a \texttt{Rot} complexity of $\left(\frac{4d}{d_0}-2d' + \frac{2d\log d_0}{d'} + 3\log d'\right)$ for $d'>1$, and $\left(\frac{3d}{d_0}+\frac{2d\log d_0}{d'} + 3\log d'\right)$ under the minimum encoding space ($d'=1$), which reduces to $3d$ $\texttt{Rot}$s for $d_0=1$.} 




At the end of this section, it is worth noting that the HMT and HMM optimized above do not represent the entirety of the depth-1 ideal decomposition search algorithm. Due to space limitations, more applications of the search algorithm are presented in Appendix C, 
where we use it to find the full-depth decomposability of another two important permutations used in Jiang \textit{et al.}'s HMM \cite{jiang2018secure}.  
\section{New Design for Arbitrary Permutation Decomposition}\label{sec: HVP}

\textcolor{black}{The idea of ideal decomposition aims to establish an upper bound for decomposition techniques. As previously discussed, this bound can indeed be achieved for specific permutations. However, not all permutations can be ideally decomposed. For example, the search algorithm struggles to find solutions for randomly generated permutations.}

The Benes network-based decomposition \cite{gentry2012fully} combined with its optimal level-collapsing strategy proposed by Halevi and Shoup \cite{halevi2014algorithms} may form a decomposition chain with multiplication depth $\log{n} - 1$ for any given permutation with length $n$, where each factor contains approximately 6 to 7 non-zero diagonals on average. \textcolor{black}{Although this can be considered as a rough approximation to full-depth ideal decompositions for permutations whose non-zero diagonals fill up the whole permutation matrix,} we aim to present a new solution outperforming such an approximation in certain scenarios.

\subsection{Overall Design}
Our design conceptualizes permutations as a directed graph network (Figure \ref{fig: network}). It consists of four major components:
\begin{itemize}[leftmargin=1.1em]
    \item \textbf{Nodes.} The nodes in the network are categorized into two types: rotation nodes and standby nodes, both taking vectors as input. Rotation nodes rotate the input vector by certain steps, while standby nodes output the input vector unchanged. 
    \item \textbf{Edges.} An edge applies a mask to the output of a node and transmits it to another node. When multiple edges converge at the same node, the vectors they carry are summed together to form the input to that node. 
    \item \textbf{Levels.} Levels correspond to the sequence in which nodes are processed. Nodes at the same level are handled concurrently. 
    \item \textbf{Groups.} The network is vertically divided into several groups, each containing a column of rotation nodes and multiple columns of standby nodes. Nodes are processed first by the order of levels, and then by groups. The network's final output is the sum of the bottom nodes' output of all groups.
\end{itemize}

\textcolor{black}{Performing a homomorphic permutation via this network structure proceeds as follows. Given a ciphertext holding the vector to be permuted, the ciphertext first enters the network at the standby node on level $0$. From there, each outgoing edge of that standby node multiplies the ciphertext by its associated mask and forwards the result to the edge’s destination node. If the destination is a standby node, it continues this process: masking and forwarding the ciphertext along its outgoing edges. If the destination is a rotation node, it first performs a rotation on the received ciphertext before passing it along via its outgoing edge. Consequently, all paths lead to the same destination at the bottom of the network, yielding the permuted ciphertext.}  

\textcolor{black}{In the following, we focus on the network construction. Given a permutation $p$ and any input vector $\mathbf{v}$, $p$ rearranges the entries of $\mathbf{v}$ to produce a new vector $\mathbf{v}' \leftarrow p(\mathbf{v})$. During the construction of the network for $p$, we let $\mathbf{v}$ enter the network (which initially consists of a single node) and continue to propagate through the network as it is incrementally built level by level. Throughout this process, we implicitly track the \textit{remaining rotation distance} for each entry of $\mathbf{v}$ as it is carried from one node to another. This distance is initialized as $r_{org}\leftarrow (j-i)\bmod{n}$, where $i$ is the entry's original position in $\mathbf{v}$, and $j$ its target position in $p(\mathbf{v})$.}

Initially, the network contains a single standby node at level $0$ and group $0$, which stores all entries of the input vector. Each entry is marked with an “unsolved” tag, indicating that it has not yet been rotated to its target position. The network is then constructed group by group. For $ i_g \geq 0 $, the $ i_g $-th group is constructed as follows.

A group starts from the minimum level containing unsolved entries. For level $\ell \geq 0$, we obtain the maximum remaining rotation distance $rot_{max}$ among all unsolved entries at this level. We set $rot$, the rotation step for the rotation node at level $\ell+1$, to be the largest power of two \textcolor{black}{no greater} than $rot_{max}$. Then we arrange each unsolved node's destination in the next level in 3 steps:

\noindent \textbf{Step 1.} We retrieve the entry’s remaining rotation distance $ r_{rem} $, original total rotation distance $ r_{org} $, and initial position $ i $ in the original vector. These values determine the entry's position in the next level's node: $i_{dst} = i+(r_{org}-r_{rem})\bmod{n}$.

\noindent \textbf{Step 2.} If $ r_{rem} \geq rot $, we assign the entry to a rotation node. So we check whether the $ i_{dst} $-th position of the rotation node at level $ \ell + 1 $ in the current group is occupied (by this point, if the rotation node has not been created, create it): 
\begin{itemize}[leftmargin=1.1em]
    \item If not occupied, we let the entry occupy this position. For the edge connecting the entry's current node and the rotation node, we set its mask's $i_{dst}$-th position to $1$. The entry's $r_{rem}$ is updated as $r_{rem}-rot$.   
    \item Otherwise, we change the entry's tag from "unsolved" to "deferred", indicating that it will go to level $\ell+1$ in a future group. 
\end{itemize}

\noindent \textbf{Step 3.} If $r_{rem}<rot$, we assign the entry to a standby node at level $\ell+1$. Based on the entry's current node, different decisions are made:
\begin{itemize}[leftmargin=1.1em]
    \item If it is a rotation node, we assign the entry to the $m+1$-th standby node where $m$ is the number of standby nodes at level $\ell$. For the edge connecting these two nodes, its mask's $i_{dst}$-th position is set to $1$.    
    \item If it is the $k$-th standby node at level $\ell$, the entry is assigned to the $k$-th standby node at level $\ell+1$. For the edge connecting these two nodes, its mask's $i_{dst}$-th position is set to $1$.
\end{itemize}
After traversing all unsolved entries, we proceed to the next level by setting $\ell\leftarrow \ell+1$ if \textcolor{black}{there exists any $r^{(i)}_{rem} > 0$ for $0 \leq i < n$, where $r^{(i)}_{rem}$ denotes the remaining rotation distance of the unsolved entry indexed $i$ in the original $\mathbf{v}$.} 
Otherwise, the $i_g$-th group's construction is completed. We change the tag of entries marked as  “unsolved” to “solved” and those marked as “deferred” to “unsolved”, and proceed to the $ i_g + 1 $-th group. \textcolor{black}{Appendix \ref{app: NWcorrect} further demonstrates the correctness of our network construction in detail.}

\begin{table*}
\ra{0.7}
\centering
\caption{Average number of rotations in different network constructions ($\log{n}$ keys)}
\begin{tabular*}{\textwidth}{@{\extracolsep{\fill}}llclc}
\toprule
Scheme & Operation & $n$ & Number of operations at network level $1\rightarrow \log{n}$ & \textcolor{black}{\#ScalarMult} \\
\midrule
Ours & Rotation & $2^{10}$ &  $\{ 1.0, 2.0, 3.3, 3.7, 4.0, 4.1, 4.3, 4.3, 4.0, 3.8\}$ & \textcolor{black}{$2.484\times 10^9$} \\ 
& (merged \texttt{Rescale})& $2^{11}$ &    $\{ 1.0, 2.0, 3.4, 3.9, 4.2, 4.3, 4.3, 4.4, 4.2, 4.1, 4.0 \}$ & \textcolor{black}{$2.814\times 10^9$}  \\ 
& & $2^{12}$ &    $\{ 1.0, 2.0, 3.6, 4.1, 4.2, 4.4, 4.6, 4.8, 4.9, 4.8, 4.7, 4.2 \}$ & \textcolor{black}{$3.101\times 10^9$} \\ 
& & $2^{13}$ &    $\{ 1.0, 2.0, 3.8, 4.4, 4.7, 5.0, 5.1, 5.2, 5.2, 5.2, 5.1, 5.0, 4.7       \}$ & \textcolor{black}{$3.489\times 10^9$} \\ 
& & $2^{14}$ &    $\{ 1.0, 2.0, 4.0, 4.7, 5.1, 5.3, 5.4, 5.4, 5.5, 5.6, 5.5, 5.4, 5.3, 5.0  \}$ & \textcolor{black}{$3.818\times 10^{9}$}  \\ 

\midrule
Benes network & Rotation & $2^{10}$ &    $\{ 4.6, 5.6, 5.6, 5.6, 4.6, 5.6, 5.5, 5.5, 5.5, 0 \}$ & \textcolor{black}{$3.791\times 10^9$} \\ 
\cite{gentry2012fully,halevi2014algorithms}& (double-hoisted) & $2^{11}$  &  $\{ 5.6, 3.0, 5.6, 5.6, 5.6, 4.6, 5.5, 5.5, 5.5, 5.4, 0 \}$ & \textcolor{black}{$3.844\times 10^9$} \\ 
& & $2^{12}$  & $\{ 4.6, 5.6, 5.6, 5.6, 4.6, 5.6, 5.5, 5.5, 5.5, 5.4, 5.4, 0 \}$ & \textcolor{black}{$4.178\times 10^9$} \\ 
& & $2^{13}$  &  $\{ 5.6, 3.0, 5.6, 5.6, 5.6, 5.6, 4.5, 5.5, 5.5, 5.4, 5.4, 5.4, 0\}$  & \textcolor{black}{$4.207 \times 10^9$} \\ 
& & $2^{14}$  &  $\{ 4.6, 5.6, 5.6, 5.6, 5.6, 5.6, 4.5, 5.5, 5.5, 5.4, 5.4, 5.4, 5.3, 0 \}$ & \textcolor{black}{$4.508\times 10^9$} \\ 
\bottomrule
\end{tabular*}
\label{tb: rotseachLv}
\end{table*}

\subsection{Technical Detail}

\textcolor{black}{We present the key features of our proposed network construction along with techniques for efficient network evaluation. }

\subsubsection{Multi-group Network Construction}\label{sec: multi-group}

The selection of $rot$ at each level aims to limit both the number of rotations in each group and the total number of rotation keys to at most $\log n$. 
This restriction may cause occupation conflicts when multiple nodes attempt to transfer entries to a single rotation node.

Therefore, new groups are created to manage conflicting entries that can’t secure a position in their original group. Specifically, for any entry encountering an occupation conflict, we track the rotations it has completed and transfer it to one of the subsequent groups using a cross-group edge. Such a strategy ensures that entries follow the paths corresponding to the binary encoding of their total rotation distance.

The number of levels is bounded by $\log{n} + 1 - i$ for any group in the network, where $i$ is the minimum source level (index) of entries it receives from previous groups. Moreover, the number of levels of groups demonstrates a descending trend since we always consider the closest right-hand group first for transferring conflicting entries.

\subsubsection{Reduction of Masks and Depth}

Many masks within our network are unnecessary in real homomorphic evaluation. They can be replaced with copy operations without compromising the correctness of the network. 
For any group with a maximum level index of $\ell$, the unnecessary masks are those on edges reaching standby nodes, except for the edges from level $\ell-2$ to level $\ell-1$. This corresponds to the yellow, blue, and purple edges in Figure \ref{fig: network}. Reducing them to copy operations allows us to eliminate a significant number of plaintext-ciphertext multiplications and reduce the multiplication depth of the overall network to at most $\log{n}-1$. 

\textcolor{black}{We further provide a level-collapsing strategy for the network. To compress the top $1\sim \ell$ levels, we apply rotations of distances $\{ i\cdot n/2^{\ell} \mid i=0,1,\dots,2^\ell-1 \}$ to the input ciphertext and assign specific masks to them to construct the inputs for all nodes at level $\ell+1$. The multiplication depth of the merged levels is reduced to $1$.}
\textcolor{black}{We can also merge the bottom $ \ell $ levels by reconstructing the nodes at level $\log{n}-\ell$: for entries sharing the same $ r_{rem} $, we use masks to extract and consolidate them into a single ciphertext, yielding $2^{\ell}$ ciphertexts in total. By rotating them according to their respective $r_{rem}$ and summing them, we obtain the result of homomorphic permutation.} 
\textcolor{black}{ Rotations in the merged level are performed with a level-order traversal of some $m$-ary tree, where $m$ is a power of $2$. This restricts the number of additional rotation keys required per collapsed level to $\frac{m-1}{\log{m}}-1$.} 


\subsubsection{Timing of Rescale}\label{sec: rescaletime}

$\texttt{Rescale}$ is the most computationally intensive operation following $\texttt{Rot}$ in our proposed network for homomorphic permutation, and its execution timing is carefully managed during the network evaluation.
Whenever a rotation node requires rescaling its input, we reduce its computational complexity by integrating the rescaling into the final submodule of the rotation: \textbf{ModDown}, which receives a ciphertext $ ct = (a', b') \in R_{PQ_{\ell}}^2 $ from the preceding submodules and outputs $ (\lfloor a' \cdot P^{-1} \rceil, \lfloor b' \cdot P^{-1} \rceil) \in R_{Q_\ell}^2 $ \cite{cheon2019full,bossuat2021efficient}. \textcolor{black}{Rescaling} a ciphertext $ ct = (c_0, c_1) \in R_{Q_\ell}^2 $ can be viewed as an instance of \textbf{ModDown} computing $ (c_0, c_1) \in R_{Q_{\ell-1}}^2 \leftarrow (\lfloor c_0 \cdot q_{\ell}^{-1} \rceil, \lfloor c_1 \cdot q_{\ell}^{-1} \rceil) $. Thus, for a ciphertext requiring both rotation and rescaling, we can substitute the final step of the rotation with a \textbf{ModDown} from $ Q_{\ell}P $ to $ Q_{\ell-1} $: $ (\lfloor a' \cdot (q_\ell P)^{-1} \rceil, \lfloor b' \cdot (q_\ell P)^{-1} \rceil) \mod Q_{\ell-1} \in R_{Q_{\ell-1}}^2 $, completing both rotation and rescaling simultaneously. \textcolor{black}{This merging strategy is tailored to our network architecture and has negligible effect on conventional homomorphic permutation evaluation via the HLT algorithm, which performs only one \texttt{Rescale} per HLT.}

\subsection{Complexity Analysis}\label{sec: HVPcomplexity}

In this section, we analyze the primary computational complexity of our proposed multi-group network for homomorphic permutations and compare it with the Benes network-based solution.

As we propose merging $\texttt{Rescale}$ with the ciphertext rotation, the remaining fundamental operations in the network are rotations and homomorphic masking, with the former dominating the complexity. We analyze the complexity of the rotation merged with rescaling, and compare it to the original version of separate routines. 
For a ciphertext in $R_{Q_{\ell+1}}^2$ and modulus chain $P$ with $\alpha$ moduli, the number of scalar multiplication for rescaling it to $R_{Q_\ell}^2$ and performing one rotation is (see Appendix \ref{app: HE} for submodules description):

\noindent \textbf{1 (Rescale):} $$2\cdot N\cdot ((\ell+1)+ (\ell+2)\cdot (\log{(N)}+1))$$

\noindent \textbf{2 (Decompose):} $$N\cdot \log{(N)}\cdot (\ell+1) +  \lceil \frac{\ell+1}{\alpha} \rceil \cdot  N \cdot  ( \alpha+(\ell+1) \cdot  (\log{(N)} + \alpha))$$

\noindent \textbf{3 (MultSum):} $$2\cdot N\cdot\lceil \frac{\ell+1}{\alpha} \rceil \cdot  (\ell+1+\alpha)$$


\noindent \textbf{4 (ModDown):} $$2\cdot N\cdot ((\ell+1)\cdot (\alpha+\log{(N)}+1) + \alpha\cdot (\log{(N)}+1))$$
On the other hand, the number of scalar multiplication for completing the merged rotation is:

\noindent \textbf{1 (Decompose):} $$N\cdot \log{(N)}\cdot (\ell+2) +  \lceil \frac{\ell+2}{\alpha} \rceil \cdot  N \cdot  ( \alpha+(\ell+2) \cdot  (\log{(N)} + \alpha) )$$

\noindent \textbf{2 (MultSum):} $$2\cdot N\cdot \lceil \frac{\ell+2}{\alpha} \rceil \cdot  (\ell+2+\alpha)$$


\noindent \textbf{3 (ModDown+Rescale):} $$2\cdot N\cdot (\ell\cdot ((\alpha+1)+\log{(N)}+1) + (\alpha+1)\cdot (\log{(N)}+1)).$$
\textcolor{black}{By subtracting the complexity of the merged version from that of the separate version, we obtain a difference no smaller than
$$N\cdot(\ell \log N + 3.5\log N - \ell + 8)$$ with $\alpha \geq 2$ under common setting, indicating a complexity reduction across all possible ciphertext modulus levels.} 





\textcolor{black}{The average complexity for evaluating random permutations of varying lengths $n$ is presented in Table \ref{tb: rotseachLv}, based on 20 samples per length. The table details the number of rotations required at each network level, from $1$ to $\log n$.}
For comparison, the corresponding information for the Benes network-based scheme \cite{gentry2012fully,halevi2014algorithms} is also provided. In their case, each level corresponds to one double-hoisted HLT \cite{bossuat2021efficient} (see Appendix \ref{app: benes} for implementation details).
\textcolor{black}{The number of scalar multiplications (\#ScalarMult) for all rotations in the network is also presented.} 
This is obtained with $\log{N}=15$, modulus chain $Q$ of $18$ moduli, and $P$ of $3$ moduli. 
At any level $\ell$, the rotations are applied to ciphertexts in $R_{Q_{17-\ell}}^2$. 

From Table \ref{tb: rotseachLv}, we can observe the following:
\textcolor{black}{(i) In our network, the number of rotations first increases and then decreases as the level index grows. The initial increase implies that new groups are created to continue rotating conflicting nodes,  
while the subsequent decrease arises from the gradual reduction in the number of entries each group needs to process.} 
(ii) For the same $n$, our network exhibits a lower rotation complexity, and this advantage is more pronounced for smaller values of $n$. For $\log{n}=10\sim 14$, the corresponding advantage is approximately $1.52\times \sim 1.18\times$. 
\textcolor{black}{The variation in the overall number of rotations remains modest in our network, with a slightly increasing standard deviation from $1.6$ to $2.4$ as the permutation length grows.} 

\begin{table*}[h]
\ra{0.85}
\setlength{\tabcolsep}{4.2pt}
\centering
\begin{tabular*}{\textwidth}{@{}lccccccccccccccccccc@{}}\toprule
  &  & 
  & \multicolumn{2}{c}{$\ell = 1$ (Mult Depth=$2$)}  & \phantom{ab} 
  & \multicolumn{2}{c}{$\ell = 2$ (Mult Depth=$3$)}  & \phantom{ab}
  & \multicolumn{2}{c}{$\ell = 3$ (Mult Depth=$4$)}  & \phantom{ab}
  & \multicolumn{2}{c}{$\ell = 4$ (Mult Depth=$5$)}  & 
\\
\cmidrule{4-5} \cmidrule{7-8} \cmidrule{10-11} \cmidrule{13-14}
Scheme & $d$ & BSGS $d1$ 
& Time (ms) & Keys &
& Time (ms) & Keys &
& Time (ms) & Keys &
& Time (ms) & Keys 
\\
\midrule
Original \cite{jiang2018secure} & 128 & $32$ &	1211 & 38 & &	1108 & 38 & &	961 & 38 & &	896  & 38 \\
Original \cite{jiang2018secure} & 128 & $16$ &	1469 & 30 & &	1437 & 30 & &	1202 & 30 & &	1111  & 30 \\
\textcolor{black}{Ours}                   & \textcolor{black}{128} & \textcolor{black}{$32$} &   \textcolor{black}{\textbf{868}} & \textcolor{black}{35} & &	\textcolor{black}{787} & \textcolor{black}{35} & &	\textcolor{black}{-} & \textcolor{black}{-} & &	\textcolor{black}{-}       & \textcolor{black}{-} \\

Ours                   & 128 & $16$ &   1018 & 23 & &	\textbf{722} & 21 & &	\textbf{598} & 21 & &	-       & - \\
Ours                   & 128 & $8$  &   1380 & 23 & &	844 & 17 & &	600 & 15 & &	\textbf{569}  & 15 \\
\midrule

\textcolor{black}{Original \cite{jiang2018secure}} & \textcolor{black}{256} & \textcolor{black}{$64$} &	\textcolor{black}{4376} & \textcolor{black}{70} &  &	\textcolor{black}{4027} & \textcolor{black}{70} & &	\textcolor{black}{3507} & \textcolor{black}{70} & & \textcolor{black}{3372}  & \textcolor{black}{70} \\
\textcolor{black}{Original \cite{jiang2018secure}} & \textcolor{black}{256} & \textcolor{black}{$32$} &	\textcolor{black}{4622} & \textcolor{black}{46} &  &	\textcolor{black}{4509} & \textcolor{black}{46} & &	\textcolor{black}{3735} & \textcolor{black}{46} & & \textcolor{black}{3609}  & \textcolor{black}{46} \\

\textcolor{black}{Ours}                   & \textcolor{black}{256} & \textcolor{black}{$64$} &   \textcolor{black}{3352} & \textcolor{black}{67} & &	\textcolor{black}{2715} & \textcolor{black}{67} & &	\textcolor{black}{-} & \textcolor{black}{-} & &	\textcolor{black}{-}       & \textcolor{black}{-} \\
\textcolor{black}{Ours}                   & \textcolor{black}{256} & \textcolor{black}{$32$} &   \textcolor{black}{\textbf{2890}} & \textcolor{black}{39} & &	\textcolor{black}{\textbf{2186}} & \textcolor{black}{37} & &	\textcolor{black}{1710} & \textcolor{black}{37} & &	\textcolor{black}{-}       & \textcolor{black}{-} \\
\textcolor{black}{Ours}                   & \textcolor{black}{256} & \textcolor{black}{$16$} &   \textcolor{black}{3766} & \textcolor{black}{31} & &	\textcolor{black}{2311} & \textcolor{black}{25} & & \textcolor{black}{\textbf{1604}}	 & \textcolor{black}{23} & &	      \textcolor{black}{\textbf{1485}} & \textcolor{black}{23} \\
\bottomrule
\end{tabular*}
\caption{HMT benchmarks under dimension $d\in \{128,256\}$ and different decomposition depth $\ell$} \label{tb: HMT2}
\end{table*}

\begin{table*}[h]
\ra{0.85}
\setlength{\tabcolsep}{4.7pt}
\centering
\begin{tabular*}{\hsize}{@{}lccccccccccccccccccc@{}}\toprule
  & \phantom{abc}
  & \multicolumn{3}{c}{$d=8$ }  & \phantom{a} 
  & \multicolumn{3}{c}{$d=16$}  & \phantom{a}
  & \multicolumn{3}{c}{$d=32$}  & \phantom{a}
  & \multicolumn{3}{c}{$d=64$}  & \phantom{a}
\\
\cmidrule{3-5} \cmidrule{7-9} \cmidrule{11-13} \cmidrule{15-17}
Scheme & $d'$ 
& $m$ & \textcolor{black}{Time (ms)} & Keys &
& $m$ & \textcolor{black}{Time (ms)} & Keys &
& $m$ & \textcolor{black}{Time (ms)} & Keys &
& $m$ & \textcolor{black}{Time (ms)} & Keys 
\\
\midrule

\textcolor{black}{\cite{jiang2018secure,MA2024103658}} & \textcolor{black}{1} 
& \textcolor{black}{512} & \textcolor{black}{1802} & \textcolor{black}{10} & 
& \textcolor{black}{128} & \textcolor{black}{3338} & \textcolor{black}{16} &
& \textcolor{black}{32} & \textcolor{black}{4973} & \textcolor{black}{28} &
& \textcolor{black}{8} & \textcolor{black}{9043} & \textcolor{black}{45} \\ 
Original \cite{rizomiliotis2022matrix} & d   
& 64 & \textcolor{black}{1697}  & 23 & 
& 8  & \textcolor{black}{3072}& 42 &
& 1  & \textcolor{black}{5530} & 77 &
& -  & -      & -  \\
Ours     & d   
& 64 & \textcolor{black}{\textbf{1046}}  & 21 &
& 8  & \textcolor{black}{\textbf{1336}}  & 28 &
& 1  & \textcolor{black}{\textbf{1705}}  & 35 &
& -  & -      & -  \\
Ours     & d/2 
& 128& \textcolor{black}{1257}  & \textcolor{black}{18} &
& 16 & \textcolor{black}{1730}  & \textcolor{black}{25} &
& 2  & \textcolor{black}{2175}  & \textcolor{black}{32} &
& -  & -      & -  \\
Ours     & d/4 
& 256& \textcolor{black}{1800}   & \textcolor{black}{15} &
& 32 & \textcolor{black}{2584}   & \textcolor{black}{22} &
& 4  & \textcolor{black}{3283}   & \textcolor{black}{29} &
& -  & -      & -  \\
Ours     & d/8 
& 512& \textcolor{black}{3354}   & \textcolor{black}{9} &
& 64 & \textcolor{black}{4507}   & \textcolor{black}{19} &
& 8  & \textcolor{black}{5736}   & \textcolor{black}{26} &
& 1  & \textcolor{black}{\textbf{7050}}   & \textcolor{black}{33} \\

Ours     & d/16
& -  & -      & -  &
& 128& \textcolor{black}{8757}   & \textcolor{black}{12} &
& 16 & \textcolor{black}{11078}  & \textcolor{black}{23} &
& 2  & \textcolor{black}{12943}  & \textcolor{black}{30} \\

Ours     & d/32
& -  & -      & -  &
& -  & -      & -  &
& 32 & \textcolor{black}{21773}  & \textcolor{black}{15} &
& 4  & \textcolor{black}{25720}  & \textcolor{black}{27} \\
Ours     & d/64
& -  & -      & -  &
& -  & -      & -  &
& -  & -      & -  &
& 8  & \textcolor{black}{51616}  & \textcolor{black}{18} \\

\bottomrule
\end{tabular*}
\caption{HMM benchmarks under full decomposition depth and constant multiplication depth}\label{tb: HMM} 
\end{table*}

\interfootnotelinepenalty=10000

\section{Implementation}

We implement the algorithms proposed in the previous sections using the full-RNS CKKS HE scheme \cite{cheon2019full,kim2019approximate} provided by the Lattigo library \cite{mouchet2020lattigo}.
Experiments are conducted on a system with an i7-13700K processor running at 3.40 GHz and a memory of 64G.
\textcolor{black}{We present benchmarks for homomorphic matrix transposition, homomorphic matrix multiplication, along with its application in oblivious neural network inference, and arbitrary homomorphic permutations. Unless otherwise specified, the homomorphic encryption parameters are fixed at $\log{QP}\approx 880$ and $\log{N}=15$ across all benchmarks. This provides a 128-bit security level \cite{albrecht2015concrete,albrecht2021homomorphic,bossuat2024security}, an encoding space of $N/2$ complex numbers per ciphertext, and a budget of 18 modulus levels, indexed from $0$ to $17$.}
When comparing different algorithms, we ensure that their execution concludes at the same ciphertext modulus level.

\subsection{Homomorphic Matrix Transposition}

Table \ref{tb: HMT2} provides HMT benchmarks under dimension \textcolor{black}{$d\in \{128,256\}$} and decomposition depth \textcolor{black}{$\ell \leq \lfloor\log{(d-1)}\rfloor$}. 
\textcolor{black}{Both the depth-$\ell$ decomposed HMT and the original version used for comparison terminate at modulus level $17 - (\ell + 1)$. To support HMT of $d=256$, the polynomial degree modulus $N$ is extended to $2^{16}$, and the encoding capacity is alternated to $N$ real numbers by the conjugate-invariant ring based CKKS \cite{kim2019approximate}.} As previously discussed, the leftmost factor $U^t_{L_\ell}$ is implemented with the BSGS algorithm. \textcolor{black}{$U^t_{L_\ell}$ contains $2d/2^\ell$ non-zero diagonals, distributed symmetrically on both sides of the $0$-th diagonal. Thus, the product of BSGS loop sizes $d_1d_2$ is no greater than $d/2^{\ell}$, the number of non-zero diagonals on one side of $U^t_{L_\ell}$.} 

Compared to the original undecomposed version \cite{jiang2018secure}, the ideal decomposition yields significant complexity reductions, even when the decomposition depth is shallow. 
\textcolor{black}{This aligns with the reduced number of rotations from $d_1 + \frac{2d}{d_1}$ to $2\ell+d_1+\frac{2d}{2^{\ell} d_1}$ by depth-$\ell$ decomposition. Notably, our method demonstrates a speed-up of up to $2.3\times$ and a rotation key reduction of up to $67\%$ for $d=256$. The asymptotic complexity improvement can also be evidenced by the runtime growth with respect to $d$. While the original $U^t$ exhibits a consistent growth factor of $3.7$, the decomposed version shows a decreasing trend from $3.3$ to $2.6$ as $\ell$ increases.} 

\subsection{Homomorphic Matrix Multiplication}

Table \ref{tb: HMM} primarily compares our enhanced HMM in Section \ref{sec: HMMDmp} to the original scheme \cite{rizomiliotis2022matrix}. Unlike our variable \( d^2d' \) matrix encoding space, the original \( d^3 \)-encoding only allows limited choices of dimensions. \textcolor{black}{Thus, the ciphertext's encoding capacity is also shifted to $N$ real numbers \cite{kim2019approximate} (HE parameters remain default). All computations terminate at modulus level $14$.} 


In the table, parameter $ m = N/d^2d' $ denotes the maximum number of matrices that can be computed in parallel. For the original encoding setting $ d' = d $, our scheme achieves \textcolor{black}{a speedup of up to $3.24\times$ and a 54\% reduction} in rotation keys compared to the original algorithm. \textcolor{black}{This improvement is purely achieved by the permutation decomposition technique, which allows the complexity to grow logarithmically with $ d $ and remain independent of multiplication depth overhead.}

\textcolor{black}{The rest of the experiments} under different $ d' $ demonstrate the enhancement of our more flexible encoding space utilization method. We can trade encoding space for higher parallelity and lower amortized time per matrix multiplication.
\textcolor{black}{The HMM algorithm proposed by Jiang \textit{et al.} is also included for reference \cite{jiang2018secure}. Here, we adopt the version enhanced by Ma \textit{et al.} through hoisting techniques \cite{MA2024103658}. Recall Table \ref{tb: HMMcomplexity}, this HMM algorithm achieves optimal amortized complexity. However, its complexity is fixed over different batch sizes of matrix multiplication. Thus, when it comes to sequential and small-batch matrix multiplications, our approach offers a solution with significantly lower latency. A representative scenario of such use cases is explored in the following subsection.}

\begin{table}
\ra{0.9}
\setlength{\tabcolsep}{4.5pt}
\begin{tabular*}{\columnwidth}{@{}lllllllllllll@{}}
\toprule
Layer & Description \\ 
\midrule
Convolution    & Input: $28\times 28$, kernel:$7\times 7$, stride:$3$, channels:$2$ \\ 
$1^{st}$ square & Squaring $128$ inputs \\ 
FC-1           & Fully connecting $128$ inputs and $32$ outputs \\ 
$2^{nd}$ square & Squaring $32$ inputs \\ 
FC-2           & Fully connecting $32$ inputs and $10$ outputs \\ 
\bottomrule
\end{tabular*}
\caption{Description of E2DM-lite CNN on the MNIST dataset}\label{tb: ONNI0}
\end{table}

\begin{table}
\ra{0.75}
\setlength{\tabcolsep}{3.6pt}
\begin{threeparttable}
\begin{tabular*}{\columnwidth}{@{}lccccccccccc@{}}
\toprule
Framework & Batch & Latency & Amortized & Depth & Accuracy \\ 
\midrule
E2DM \cite{jiang2018secure}     & 64 & 4232 ms & 66 ms & 9 & 98.1\% \\ 
CDKS \cite{chen2019efficient}   & 1  & 308 ms  & 308 ms  & 6 & 98.4\% \\ 
POSEIDON \cite{sav2021poseidon} & 1 & 155 ms & 155 ms  & 6 & 89.9\% \\
E2DM-lite\tnote{*} & 128 & 2795 ms & 22 ms& 9 & 97.7\% \\
\textcolor{black}{E2DM-lite}\tnote{*} & \textcolor{black}{16} & \textcolor{black}{1839 ms} & \textcolor{black}{114 ms} & \textcolor{black}{9} & \textcolor{black}{97.7\%} \\
\textcolor{black}{E2DM-lite}\tnote{*} & \textcolor{black}{8} & \textcolor{black}{1414 ms} & \textcolor{black}{176 ms} & \textcolor{black}{9} & \textcolor{black}{97.7\%} \\
E2DM-lite\tnote{\textdagger} & 32 & 917 ms& 28 ms& 6 & 97.7\% \\
E2DM-lite\tnote{\textdagger} & 16 & 523 ms& 33 ms& 6 & 97.7\% \\
E2DM-lite\tnote{\textdagger} & 8 &  363 ms& 45 ms& 6 & 97.7\% \\
\bottomrule
\end{tabular*}
\begin{tablenotes}
  \item[*] Original implementation under Jiang \textit{et al.}'s HMM \cite{jiang2018secure,MA2024103658}.
  \item[\textdagger] Implemented with our HMM.
\end{tablenotes}
\end{threeparttable}
\caption{MNIST benchmarks of HE-based ONNI schemes}\label{tb: ONNI2}
\end{table}

\begin{table}
\ra{0.8}
    \setlength{\tabcolsep}{2.7pt}
    \begin{tabular*}{\columnwidth}{@{}cccccccc@{}}
    \toprule
    \phantom{$\log{n}$}   &  & 
    \multicolumn{2}{c}{\cite{gentry2012fully,halevi2014algorithms}}  &  & 
    \multicolumn{3}{c}{Ours}  \\
    \cmidrule{3-4} \cmidrule{6-8}
    
    $\log{n}$      &  &  
    Time (ms)       &  
    \textcolor{black}{Collapsed}      &  &
    Time (ms)       &
    \textcolor{black}{Collapsed}      & 
    Speed-up
\\ 
\midrule
$7$  & & \textcolor{black}{308} & \textcolor{black}{3} & & \textcolor{black}{187}  & \textcolor{black}{Top 0 Bot 2} & 
\textcolor{black}{$\times 1.65$} \\
$8$  & & \textcolor{black}{426} & \textcolor{black}{3} & & \textcolor{black}{252}  & \textcolor{black}{Top 0 Bot 2}  & $\times 1.69$ \\
$9$  & & \textcolor{black}{553} & \textcolor{black}{5} & & \textcolor{black}{362} & \textcolor{black}{Top 0 Bot 3}  & $\times 1.53$ \\
$10$ & & \textcolor{black}{655} & \textcolor{black}{4} & & \textcolor{black}{481} & \textcolor{black}{Top 0 Bot 3}  & $\times 1.36$ \\
$11$ & & \textcolor{black}{834} & \textcolor{black}{5} & & \textcolor{black}{617} & \textcolor{black}{Top 2 Bot 3}  & $\times 1.35$ \\
$12$ & & \textcolor{black}{1013} & \textcolor{black}{5} & & \textcolor{black}{810} & \textcolor{black}{Top 2 Bot 3}  & \textcolor{black}{$\times 1.25$} \\
$13$ & & \textcolor{black}{1184} & \textcolor{black}{4} & & \textcolor{black}{1030} & \textcolor{black}{Top 2 Bot 3}  & $\times 1.15$ \\
$14$ & & \textcolor{black}{1405} & \textcolor{black}{6} & & \textcolor{black}{1244} & \textcolor{black}{Top 2 Bot 3}  & \textcolor{black}{$\times 1.13$} \\
\bottomrule
\end{tabular*}
\caption{Average time consumption for randomly sampled permutations of length $n$ ($\log{n}$ rotation keys).}\label{tb: HVP}
\end{table}

\subsection{Application to Oblivious Neural Network Inference}

Oblivious inference refers to the task of outsourcing an ML model without revealing its parameters or inputs, wherein the model owner is not required to run the model locally \cite{jiang2018secure, chen2019efficient, rizomiliotis2022partially}. Jiang \textit{et al.} \cite{jiang2018secure} proposed an efficient HE-based framework, E2DM, to realize this objective, predicting the MNIST dataset \cite{deng2012mnist} over an encrypted neural network. 
The computational complexity \textcolor{black}{of this framework} is dominated by their proposed HMM. 
We adopt their service paradigm but replace the HMM module in their framework with our version. Our objective is to demonstrate the practicality of our HMM by achieving lower latency when redundant encoding capacity is available, particularly under small-batch inference scenarios.
Table \ref{tb: ONNI0} describes the employed neural network. 
We call it E2DM-lite, as it is derived by halving the number of convolutional channels and neurons in E2DM, but achieves an accuracy of 97.7\% (only 0.4\% lower than the original version). 
We present benchmarks of E2DM-lite in Table \ref{tb: ONNI2} (implementation detail provided in Appendix \ref{app: ONNI}). When implemented with the original HMM by Jiang \textit{et al.}, E2DM-lite requires an inference batch size of $128$ for optimal performance. \textcolor{black}{Notably, its latency remains unchanged when the batch size drops to $32$, as the framework utilizes parallel HMMs of size $32 \times 32$ to process fully connected layers. When the batch size further decreases to $16$ or $8$, the parallel HMM blocks become smaller, and the cost of aggregating these blocks after HMM increases correspondingly.} 
\textcolor{black}{In contrast, our HMM takes advantage of the encoding space made available by small batches, achieving latency improvements of $3.0\times$, $3.5\times$, and $3.9\times$ for batch sizes of 32, 16, and 8, respectively, while maintaining an amortized runtime that is only slightly higher than that of \cite{jiang2018secure}.} 
\textcolor{black}{Note that although it is possible to implement the original E2DM with our proposed techniques, the HMM size in E2DM is defined relatively larger, where using a redundant encoding space for latency reduction would result in increased polynomial degree modulus $N$, posing a negative effect on the overall performance of homomorphic operations.}
 
Table \ref{tb: ONNI2} also presents benchmarks of other state-of-the-art HE-based oblivious neural network inference frameworks on the MNIST dataset. Note that although they use similar underlying HE schemes \cite{cheon2017homomorphic,cheon2019full} and network topologies, we cannot force them to work under identical parameters of HE and neural network, which have been highly customized for different computation strategies. Therefore, the benchmarks are obtained under the original parameters of these frameworks, except for E2DM, which only ensures 80-bit security. We have it implemented under the same HE setting as our scheme. 

Jiang \textit{et al.}'s HMM implementation in E2DM and E2DM-lite exhibits good amortized performance under high batch sizes, but introduces relatively high depth, making it more suitable for shallow, wide networks. 
The strategy designed by Chen \textit{et al.}\cite{chen2019efficient} (CDKS) employs their proposed multi-key homomorphic encryption to avoid introducing a third-party authority. It has a good latency, but the batch size is relatively small. POSEIDON\cite{sav2021poseidon} is a federated neural network learning framework based on threshold HE. It provides fast oblivious inference after the network is homomorphically trained among parties. However, clients outside the training system need to encrypt their inquiries with the joint public key of the system. 

\subsection{Arbitrary Homomorphic Permutation}

Table \ref{tb: HVP} compares our proposed homomorphic permutation algorithm in Section \ref{sec: HVP} with the current optimal approach based on the Benes network \cite{gentry2012fully,halevi2014algorithms}.
The table reports the average performance on randomly generated permutations of various lengths $n$.
\textcolor{black}{The homomorphic evaluation for all permutation lengths is configured to terminate at modulus level $0$. Both schemes initially constructs networks with $\log{n}-1$ multiplication depth, then apply their respective level-collapsing strategies to seek optimal performance, where the reduced number of network levels is reported in the table.} 

The Benes network-based solution employs the level-collapsing strategy proposed by Halevi \textit{et al.}\cite{halevi2014algorithms}. For our network, the best practice is to compress the top $2\sim3$ levels and the bottom $3\sim 4$ levels. Compressing these levels reduces the number of homomorphic operations performed on relatively high modulus levels. Rotations for the compressed levels are further accelerated with the hoisting technique.
The Speed-up column in Table \ref{tb: HVP} reveals that our algorithm provides speed increase across different $n$. 
Given the typical setup range of $10\leq \log{N} \leq 16$ and $n\leq N$ for RLWE-based batch-encoding homomorphic encryption schemes \cite{bossuat2024security}, our algorithm demonstrates an efficiency advantage over nearly all practical choices of $n$.

\section{Conclusion}
In this paper, we employed more effective decompositions on permutations to enhance homomorphic permutation. We proposed two approaches to achieve this: (i) Analyzing and proving the full-depth ideal decomposability of permutations based on our proposed depth-1 ideal decomposition search algorithm. (ii) Designing a new network structure for performing a new form of decomposition and homomorphic computation on arbitrary permutations.

\begin{acks}
We extend our gratitude to all those who provided valuable insights and suggestions for this paper. 
\end{acks}


\bibliographystyle{ACM-Reference-Format} 
\bibliography{sample}

\appendix


\section{Batch-encoding HE Settings}\label{app: HE}
A batch-encoding HE scheme generally has the following setup routines:
\begin{itemize}[leftmargin=1.1em]
    \item $\texttt{Setparams}(\cdot)$: Generate parameters $params$: $(N, Q, P, \chi_{sec}, \chi_{err})$. $N$ is a power of two, which along with $Q$ and $P$ defines the polynomial rings $R = \mathbb{Z}[X]/(X^N + 1)$, $R_Q = \mathbb{Z}_Q[X]/(X^N + 1)$, and $R_{QP} = \mathbb{Z}_{QP}[X]/(X^N + 1)$, where $Q = \prod_{i=0}^{L-1} q_i$, $P = \prod_{i=0}^{\alpha-1} p_i$, and any $q_i$ for $0 \leq i < L$ or $p_i$ for $0 \leq i < \alpha$ are congruent to $1$ modulo $2N$. The distributions $\chi_{sec}$ and $\chi_{err}$, defined over $R$, are used for sampling the secret key and noise, respectively. All subsequent routines implicitly receive $params$ as input.
    
    \item $\texttt{KeyPairGen}(\cdot)$: Generate the secret key $s \leftarrow \chi_{sec}$ and compute the public key $(a, -as + e_1)$, where $a \leftarrow_U R_Q$.

    \item $\texttt{SwitchKeyGen}(s,s',\mathbf{w})$ For the input integer decomposition basis of $\beta$ elements, output $swk_{s\leftarrow s'}\in R_{Q_LP}^{\beta \times 2}$, where $swk_{s\leftarrow s'}[i] = (a_i,-a_is'+sw^{(i)}P+e_i)$ and $a_i\leftarrow_U R_{Q_LP},e_i\leftarrow \chi_{err}$. 
    
    \item $\texttt{Enc}(m, \ell)$: For the given plaintext $m \in R$ and the index $\ell$ of the modulus chain $Q$, compute $(u, v) \in R_{Q_\ell}^2 \leftarrow (ar + e_2, (-as + e_1)r + \Delta_1 m + \Delta_2 e_3)$, where $Q_\ell = \prod_{i=0}^{\ell-1} q_i$ and $e_1, e_2, e_3 \leftarrow \chi_{err}$. $\Delta_1$ and $\Delta_2$ are scaling factors used to control the precision of the plaintext, where one of them must be $1$, while the other scales the polynomial coefficients, covering both low-bit and high-bit encoding cases for $m$.
    
    \item $\texttt{Dec}(c, s)$: For the given ciphertext $c = (u, v) \in R_{Q_\ell}^2$ and secret key $s$, compute $\Delta_1 m + \Delta_2 e \leftarrow \left\langle (u, v), (s, 1) \right\rangle \in R_q$. With appropriately chosen $\Delta_1$ and $\Delta_2$, one can perform scaling, rounding, and reduction to approximately extract $m$ as the output. For example, setting $\Delta_1 m = \lfloor \frac{q}{t} m \rfloor$ and $\Delta_2 = 1$ \cite{fan2012somewhat}, we obtain $m \in R_t \leftarrow \lceil \frac{t}{q} \left\langle (u, v), (s, 1) \right\rangle \rfloor \mod{t}$, where $t$ is the plaintext coefficient moduli.
\end{itemize}
Rotating a ciphertext $ct = (c_0, c_1) \in R_{Q_\ell}^2$ with rotation key $rtk_k\in R_{PQ_L}^{\beta \times 2} \leftarrow \texttt{SwitchKeyGen}(\phi_k(s),s,\mathbf{w})$ requires the following steps \cite{bossuat2021efficient}, where $\phi_k$ is the automorphism for rotating the plaintext vector in $ct$ by $k$ steps:
\begin{enumerate}[leftmargin=1.7em]
    \item \textbf{Decompose}: $c_1$ is switched out of the NTT domain and decomposed based on basis $\mathbf{w}$ of $\beta$ elements. Each decomposed component is then switched back into the NTT domain, forming the vector $\mathbf{d}$.
    \item \textbf{MultSum}: The vector $\mathbf{d}$ is inner-multiplied with $rtk_k$, yielding the tuple $(a, b) \in R_{PQ_\ell}^2$.
    \item \textbf{Permute}: The tuple $(c_0 P + a, b)$ undergoes an automorphism $\phi_k$, resulting in $(a', b') \in R_{PQ_\ell}^2$.
    \item \textbf{ModDown}: The computation $(\lfloor a' \cdot P^{-1} \rceil, \lfloor b' \cdot P^{-1} \rceil) \in R_{Q_\ell}^2$ is performed.
\end{enumerate}
We use $\beta=\lceil (\ell+1)/\alpha \rceil$ for the complexity analysis in Section \ref{sec: HVPcomplexity}, which is aligned with the RNS decomposition basis in \cite{bossuat2021efficient}. 


\section{Section \ref{sec: DMP4HMTHMM} Supplemental Materials}

\subsection{Partition pattern of arbitary dimension}\label{thm: Only2Sizes}

\begin{theorem}
    Let $ A $ be any square matrix of size $ d \times d $ and its initial block size is $ d \times d $. Define one round of block partitioning as follows: each block of $ A $ with an odd dimension is partitioned and recombined according to equations $\ref{eq: A_divide}$, while each block with an even dimension is partitioned according to equation $\ref{eq: A_org}$. After any number of rounds of partitioning, $ A $ contains blocks of at most two different sizes which differ by only $1$.
\end{theorem}
\begin{proof}
    Suppose that after a certain round of partitioning, the blocks in $ A $ have dimensions differing by 1, specifically $ d' $ and $ d' + 1 $. After the next round of partitioning, if $ d' $ is odd, the blocks of dimension $ d' $ are divided into blocks of dimensions $ \frac{d' + 1}{2} - 1 $ and $ \frac{d' + 1}{2} $, while the blocks of dimension $ d' + 1 $ are divided into blocks of dimension $ \frac{d' + 1}{2} $. If $ d' $ is even, the blocks of dimension $ d' + 1 $ are divided into blocks of dimensions $ \frac{d'}{2} + 1 $ and $ \frac{d'}{2} $, while the blocks of dimension $ d' $ are divided into blocks of dimension $ \frac{d'}{2} $. Given that the two block dimensions initially produced by Equations $\ref{eq: A_divide}$ differ by only 1, by mathematical induction, after each round of partitioning, the blocks in $A$ always have only two different dimensions, differing by only 1. 
\end{proof}

\begin{table*}
\ra{0.9}
\setlength{\tabcolsep}{3.7pt}
\centering
\begin{tabular*}{\hsize}{@{}lccccccccccccccccccc@{}}\toprule
  & \phantom{abc}
  & \multicolumn{4}{c}{$d=8$ }  & \phantom{a} 
  & \multicolumn{4}{c}{$d=16$}  & \phantom{a}
  & \multicolumn{4}{c}{$d=32$}  & \phantom{a}
\\
\cmidrule{3-6} \cmidrule{8-11} \cmidrule{13-16} 
Scheme & $d_0$ 
& $m$ & Amortized & Keys & Speed-up &
& $m$ & Amortized & Keys & Speed-up &
& $m$ & Amortized & Keys & Speed-up 
\\
\midrule
Ours (\textcolor{black}{Original})
         & $d/1$  & 512  & 7 ms  & \textcolor{black}{9} & $\times 1.00$ & 
                & 128  & 72 ms & \textcolor{black}{12} & $\times 1.00$ &
                & 32   & 694 ms& \textcolor{black}{15} & $\times 1.00$ \\

Ours (Fast Rep)     
         & $d/2$   & 512  & 5 ms & \textcolor{black}{9} & $\times 1.33$ &
                 & 128  & 56 ms & \textcolor{black}{12} & $\times 1.28$ &
                 & 32   & 567 ms& \textcolor{black}{15} & $\times 1.22$ \\

Ours (Fast Rep)     
         & $d/4$   & 512  & 4 ms  & \textcolor{black}{12} & $\times 1.76$ &
                 & 128  & 42 ms & \textcolor{black}{15} & $\times 1.72$ &
                 & 32   & 471 ms& \textcolor{black}{18} & $\times 1.47$ \\

Ours (Fast Rep)                
         & $d/8$   & 512 &  \textcolor{black}{2 ms}  & \textcolor{black}{21} & \textcolor{black}{$\times 3.13$} &
                 & 128 & 30 ms & \textcolor{black}{24} & $\times 2.35$ &
                 & 32  & 354 ms& \textcolor{black}{27} & $\times 1.96$ \\

Ours (Fast Rep)               
         & $d/16$  & - & -     & -  & - &
                 & 128 &  \textcolor{black}{21 ms} & \textcolor{black}{45} & \textcolor{black}{$\times 3.40$} &
                 & 32  & 279 ms& \textcolor{black}{48} & $\times 2.49$ \\

Ours (Fast Rep)               
         & $d/32$  & - & -     & -   & - &
                 & -   & -     & -   & - &
                 & 32    &  \textcolor{black}{232 ms} &\textcolor{black}{93} & \textcolor{black}{$\times 2.99$} \\
\bottomrule
\end{tabular*}
\caption{Performance of HMM with optimized replication \textcolor{black}{(using $\ell=1,d=d_1d_0$)}}\label{tb: HMM_app}
\end{table*}

\subsection{Demonstration of $1$-padding on $U^\gamma$ and $U^\xi$}\label{app: 1-pad}

\begin{theorem}\label{thm: 1-pad}
    For any $d^2\times d$ matrix $\hat{A}$ and $d\times d^2$ matrix $\hat{B}$, $U^\gamma_{pad}(\hat{A})$ has the same first column as $U^\gamma(\hat{A})$, and $U^\xi_{pad}(\hat{B})$ has the same first $d$ columns as $U^\xi(\hat{B})$.
\end{theorem}
\begin{proof}
\textcolor{black}{Recalling Section \ref{sec: DMP4HMTHMM}, $U^\gamma$ represents a specific instantiation of $U^t$, wherein the non-zero diagonal indi
ces form an arithmetic sequence $\{-(d^2 - 1) \cdot i \mid -d < i < d\}$. However, given that the input matrix $\hat{A}$ to $U^\gamma$ has all-zero entries from row $d$ to row $d^2 - 1$, the set of non-zero diagonals is reduced to $\{-(d^2 - 1) \cdot i \mid 0 \leq i < d\}$.}
\textcolor{black}{$U^\gamma_{pad}(\hat{A})$ is constructed by padding these non-zero diagonals with $1$s, thereby rotating the entries of $\hat{A}$ by $-i(d^2 - 1)$ positions for $0 \leq i < d$ and summing the outcomes of these rotations. By treating $d$-dimensional column vectors in $\hat{A}$ as basic units, the shift by $-i(d^2 - 1)$ ensures that the $i$-th unit in row $0$ is precisely rotated to row $i$ and column $0$, where other positions of the $0$-th column remain zeros. Thus, the accumulation of these rotations yields the $0$-th column of $U^\gamma_{pad}(\hat{A})$ identical to that of $U^\gamma(\hat{A})$.}

\textcolor{black}{$U^\xi$ is another an instance of $U^t$, whose nonzero digonal indices are in the arithmetic sequence $\{-d(d-1) \cdot i \mid -d < i < d\}$. Similarly, this set of nonzero diagonals is reduced to $\{-d(d-1) \cdot i \mid 0\leq i < d\}$ as the input matrix $\hat{B}$ also contains only zeros from row $1$ to $d$ as a $d\times d^2$ matrix. $U^\xi_{pad}(\hat{B})$ is constructed by padding these non-zero diagonals with $1$s, rotating the entries of $\hat{B}$ by $-d(d-1)\cdot i$ positions for $0\leq i <d$ and aggregating these rotations. By treating $d$-dimensional row vectors in $\hat{B}$ as basic units, the shift by $-d(d-1)\cdot i$ ensures that the $i$-th unit in row $0$ is rotated to row $i$ and column $0$, where other positions of the $0$-th column remain zeros. Thus, summing all such rotations yields the $0$-th column in $U^\xi_{pad}(\hat{B})$ identical to that of $U^\xi(\hat{B})$.}
\end{proof}

\textcolor{black}{Theorem \ref{thm: 1-pad} demonstrates the correctness of the $1$-padding version of $U^\gamma$ and $U^\xi$. We now focus on the $1$-padding of their respective decompositions. Following the decomposition pattern of $U^t$, the decomposition of $U^\gamma$ is typically given as:}
\begin{equation} \label{eq: U^gammaDmp}
\left[
\begin{matrix}
    {\hat{A}'_{0,0}} & O \\
    {\hat{A}'_{0,1}} & O \\
\end{matrix}
\right]
\leftarrow U^\gamma_{R_1}(\hat{A}), 
\left[
\begin{matrix}
    U^\gamma({\hat{A}'_{0,0}}) & O   \\
    U^\gamma({\hat{A}'_{0,1}}) & O   \\
\end{matrix}
\right]
\leftarrow U^\gamma_{L_1}\left(U^\gamma_{R_1}(\hat{A})\right),
\end{equation}
\textcolor{black}{where $O$ denotes the matrix with only zeros. Let $U^\gamma_{pad_{R_1}}$ be the matrix obtained by filling all non-zero diagonals of $U^\gamma_{R_1}$ with $1$. $U^\gamma_{pad_{R_1}}(\hat{A})$ performs a rotation of step $-(d^2-1)\cdot d/2 \mod{d^3}$ on $\hat{A}$ and adds the rotated results to the original $\hat{A}$. This yields the structure:}
\begin{equation} \label{eq: U^gammaDmpPad}
\left[
\begin{matrix}
    {\hat{A}'_{0,0}} & * \\
    {\hat{A}'_{0,1}} & * \\
\end{matrix}
\right]
\leftarrow U^\gamma_{pad_{R_1}}(\hat{A}), 
\end{equation}
\textcolor{black}{where "*" indicates the block containing zeros and meaningless values. Similarly, for $1\leq \ell <\log{d}$, a $U^\gamma_{pad_{R_\ell}}$ rotates its input with distance $-(d^2-1)\cdot d/2^{\ell} \mod{d^3}$ and adds the rotated result to the original input. The behaviour of such $U^\gamma_{pad_{R_\ell}}$ is interpreted as follows: it evenly partitions each block of the input matrix into $4$ subblocks and transposes the top two subblocks within each leftmost block into the first column, while the remaining subblocks are left with either zeros or garbage values. Note that the initial block is the input matrix $\hat{A}$ itself. Then, $U^\gamma_{pad_{L_\ell}}$ is defined as performing $U^\gamma_{L}$ on each leftmost $d^2/2^\ell\times d/2^\ell$ block of its input while zeroing out all other positions. This operation costs $1$ multiplication depth.}

\textcolor{black}{The $1$-padding version of the decomposed $U^\xi$ is defined exactly the same way as $U^\gamma$, except that the size of the input matrix is $d\times d^2$, whereas the size of the input matrix of $U^\gamma$ is considered as $d^2\times d$.}  

\begin{figure}
    \centering
    \includegraphics[width=0.45\textwidth]{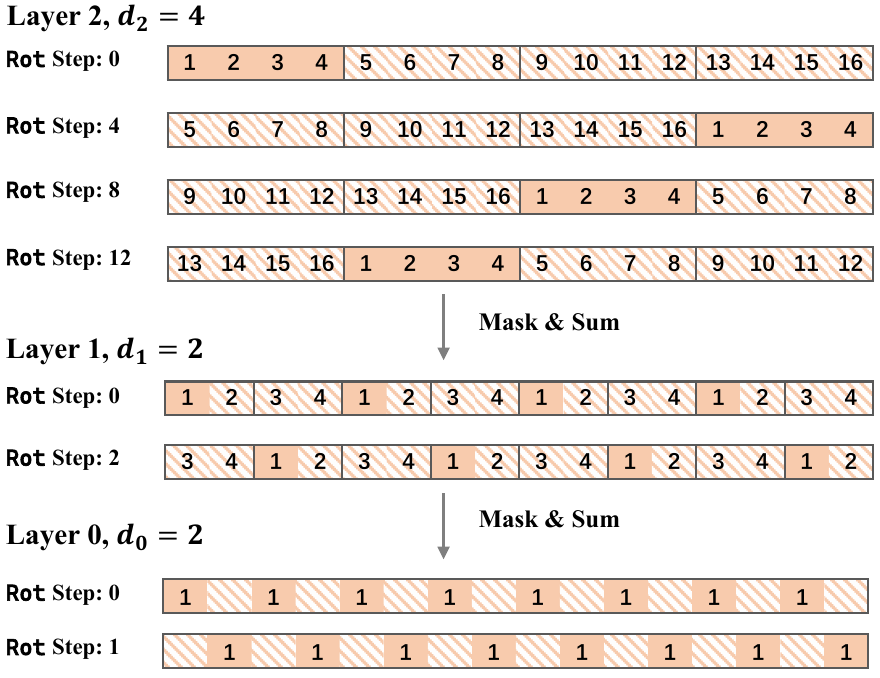}
    \caption{Optimized ciphertext replication by multi-layer construction (using $\ell=2,d=4\times 2\times 2$ as an example)}
    \label{fig: HMM_app}
\end{figure}

\subsection{Optimization of Replication}\label{app: replication}

In Section \ref{sec: HMMDmp_space}, we explore providing HMM with a configurable encoding space of $d^2d'$. As $d'$ decreases, the rotations for replication increases. We propose a method to reduce this complexity.
Consider $M = [\mathbf{m}_{0}, \mathbf{m}_{1}, \ldots, \mathbf{m}_{d-1}]$. To fill the entire $M$ with some $\mathbf{m}_{k}$, a single element replication function $\texttt{SRep}(\mathbf{u}_k\odot M;d)$ is typically applied, defined as follows:
\begin{equation*}
\begin{aligned}
& \tilde{M}_k = \left[\mathbf{0}, \ldots, \mathbf{0}, \mathbf{m}_{k}, \mathbf{0}, \ldots, \mathbf{0},\right] \leftarrow \mathbf{u}_k \odot M, \\
& \tilde{M}_k = \left[\mathbf{m}_{k}, \ldots, \mathbf{m}_{k}\right] \leftarrow  \texttt{SRep}(\tilde{M}_k; d): \\
& \tilde{M}_k \leftarrow \tilde{M}_k + \texttt{Rot}(\tilde{M}_k;i), \text{for} \phantom{a} \log{d}-1 \geq i \geq 0.
\end{aligned}
\end{equation*}
This method requires $O(t \log d)$ rotations to replicate $t$ elements. To improve efficiency, we decompose $d = d_1 d_0$ and perform rotations with step sizes at least $d_0$ prior to masking:
\begin{equation}\label{eq: replication}
\begin{aligned}
&\tilde{M}_k \leftarrow \texttt{SRep}\left(\sum^{d_1-1}_{i=0}\mathbf{u}_{k,i} \odot \texttt{Rot}\left(M;d_0 i\right);d_0\right),
\end{aligned}
\end{equation}
where $\mathbf{u}_{k,i}$ extracts $\mathbf{m}_k$ from $\texttt{Rot}\left(M; d_0 i\right)$. The pre-rotated components $\texttt{Rot}(M;d_0i)$ can be cached for all $d$ elements' computation but only costs $d_1$ rotations. The total number of rotations is $\left(\frac{d}{d_0} + d \log{d_0}\right)$, and the multiplication depth remains $1$. 
For HMM regarding $ M $ as a matrix, this method can directly apply to row-wise replication. Extending the pre-rotated steps to $ \{\pm d_0 i \mid 0 \leq i < d_1\} $ enables column-wise replication with a number of $ \left(\frac{2d}{d_0} + d \log d_0\right) $ rotations. Therefore, the number of rotations for our HMM in Section \ref{sec: HMMDmp_space} is reduced to $\left(\frac{3d}{md_0}+\frac{2d}{m}\log{d_0}\right)$ for $d'=1$.  
From an asymptotic perspective, this rotation complexity improves upon a range of schemes with $O(d \log d)$ complexity under $d^2$ encoding space \cite{cheon2018multi, huang2023secure, aikata2024secure}.
The corresponding performance, along with a comparison to the baseline version without this optimization, is presented in Table \ref{tb: HMM_app}.




This method can be extended to a multi-layer operation with depth $\ell$.
%
Let $d = d_{\ell} \ldots d_1 d_0$ and $p_i = \prod_{j=0}^{i-1}d_j$, the equation \ref{eq: replication} can be re-written by:
\begin{equation}
\begin{aligned}
&\tilde{M}_{\lfloor k/p_\ell \rfloor}^{(\ell)} 
\leftarrow \sum^{d_\ell-1}_{i=0}\mathbf{u}^{(\ell)}_{\lfloor k/p_\ell \rfloor,i} \odot \texttt{Rot}\left(M; p_\ell i\right), \\
&\tilde{M}_{\lfloor k/p_j \rfloor}^{(j)} 
\leftarrow \sum^{d_j-1}_{i=0}\mathbf{u}^{(j)}_{\lfloor k/p_j \rfloor,i} \odot \texttt{Rot}\left(\tilde{M}_{\lfloor k/p_{j+1} \rfloor}^{(j+1)}; p_j i\right), \\
&\tilde{M}_{k}^{(1)} 
\leftarrow \sum^{d_1-1}_{i=0}\mathbf{u}^{(1)}_{k,i} \odot \texttt{Rot}\left(\tilde{M}_{\lfloor k/p_2 \rfloor}^{(2)}; p_1 i\right), \\
&\tilde{M}_k
\leftarrow \texttt{SRep}\left(\tilde{M}_k^{(1)};d_0\right).
\end{aligned}
\end{equation}
where $\ell > j > 1$. For the $i$-th layer ($0 < i$), $d_{\ell}\ldots d_{i+1} (d_i - 1)$ rotations are required, yielding a rotation complexity of $\left(\frac{d}{d_0} + d \log{d_0}\right)$ in total.
A visualization of the computation process, using \( d = 4 \times 2 \times 2 = 16 \) as an example, is shown in Figure \ref{fig: HMM_app}.
Although this extension does not further reduce rotation complexity across depths $\ell$, decomposing $d$ into appropriate factors $d = d_{\ell} \ldots d_1 d_0$ reduces the number of rotation keys and multiplications. For instance, when $d_0 = 1$, the minimum rotational complexity $O(d)$ is achieved, and if $\ell = 1$, the rotation keys will be $O(d)$, but the complexity of $\texttt{CMult}$ will increase to $O(d^2)$. However, if $\ell > 1$ and $d_i > 1$ for $0 < i \leq \ell$, the rotation keys are reduced to $O(d_{\ell} + \ldots + d_2 + d_1)$. \textcolor{black}{Meanwhile, $d$ ciphertext is generated at the lowest layer ($1$-th layer) and each ciphertext requires $d_1$ $\texttt{CMult}$; $\prod_{j=i}^{\ell}d_j$ ciphertext is generated at the $i$-th layer and each ciphertext requires $d_i$ $\texttt{CMult}$; therefore $\texttt{CMult}$ totals $dd_1+\sum_{i=2}^{\ell}(d_i\prod_{j=i}^{\ell}d_j)$.}




\section{Full-depth Ideal Decomposability of Permutations in HMM by \cite{jiang2018secure}} \label{app: jiangdmp}

Jiang \textit{et al.} introduced two important permutations in their design of the HMM: $U^\sigma$ and $U^\tau$. As illustrated in Figures \ref{fig:jiangdmp}, $U^\sigma$ transforms the diagonals of the input matrix into columns of the matrix, while $U^\tau$ converts the diagonals into rows. Hence, we also refer to $U^\sigma$ as the Diag-to-Col transformation and $U^\tau$ as the Diag-to-Row transformation. 
These two permutations account for \( 5\sqrt{d} \) rotations in the HMM (which has a total of \( 3d + 5\sqrt{d} \) rotations). The full-depth ideal decomposition may reduce this part of complexity to \( 3\log{d} \) rotations. 
We also notice that these two permutations have potential applications in bridging homomorphic matrix operations that employ different encoding schemes. Some matrix operation schemes adopt diagonal encoding, which typically maps a single diagonal (using either a single entry or a block as the basic unit) of the matrix into a single ciphertext \cite{halevi2014algorithms,huang2021more}. If the dataset is initially partitioned and encrypted under row/column-major order, \( U^\sigma \) and \( U^\tau \) can be used to reorganize the entries along the diagonal into a row/column, facilitating the extraction of different diagonals into separate ciphertexts.

%

\begin{figure}[h]
    \centering
    \includegraphics[width=1\linewidth]{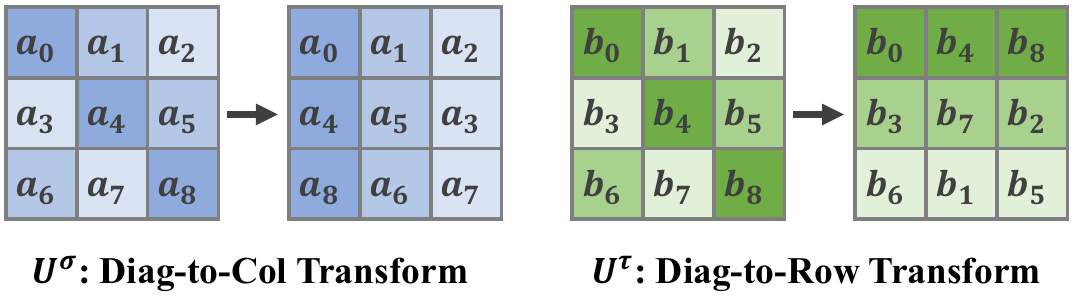}
    \caption{Demonstration of $U^\sigma$ and $U^\tau$}
    \label{fig:jiangdmp}
\end{figure}

\begin{figure*}[t]
    \centering
        \subfloat[$16\times 16$ $U^\sigma$ with non-zero entries painted in dark blue.]{\includegraphics[width=0.5\columnwidth]{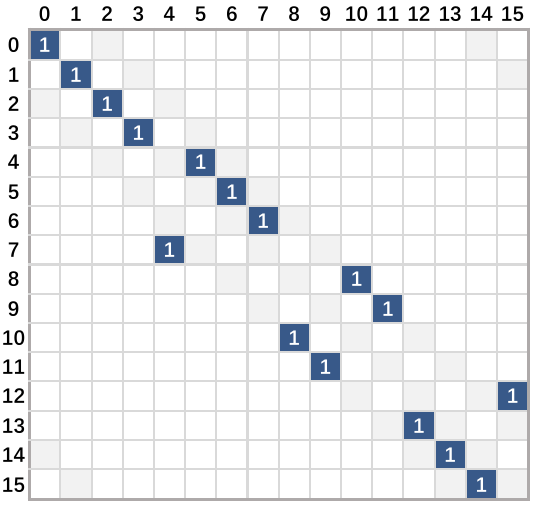}}
        \hfil 
        \subfloat[Step 1. Red entries denote the non-zero entries in $U^\sigma_R$ after Step 1.]{\includegraphics[width=0.5\columnwidth]{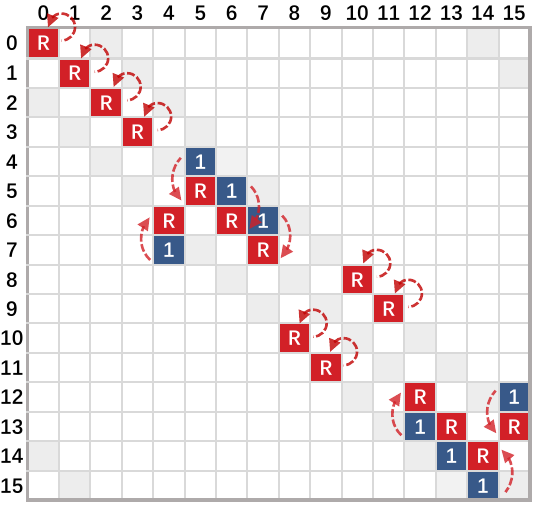}}
        \hfil
        \subfloat[Step 2. $U_R^\sigma$'s final view. Conflicting entries moved away are painted in grey]{\includegraphics[width=0.5\columnwidth]{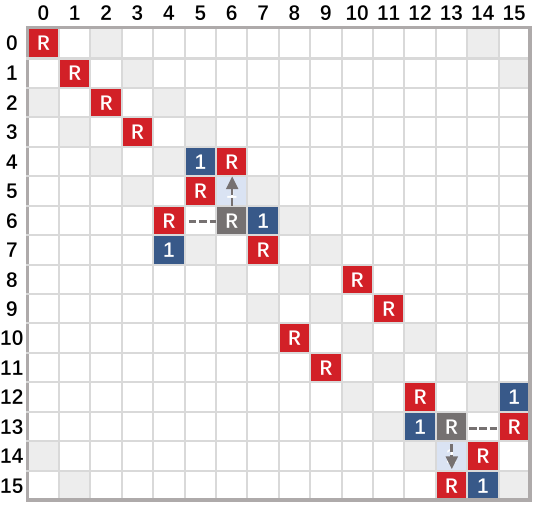}}
        \hfil 
        \subfloat[Step 3. $U^\sigma_L$ reconstructed by $U^\sigma_R$ and $U^\sigma$.]{\includegraphics[width=0.5\columnwidth]{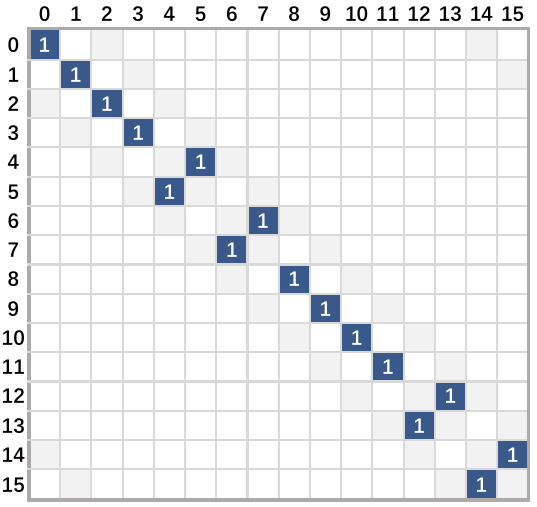}}
    \caption{Applying depth-1 ideal decomposition search to $U^\sigma$ of size $16\times 16$.}
    \label{fig: depth-1DmpSigma}
\end{figure*}

\begin{figure}[t]
    \centering
    \includegraphics[width=1\linewidth]{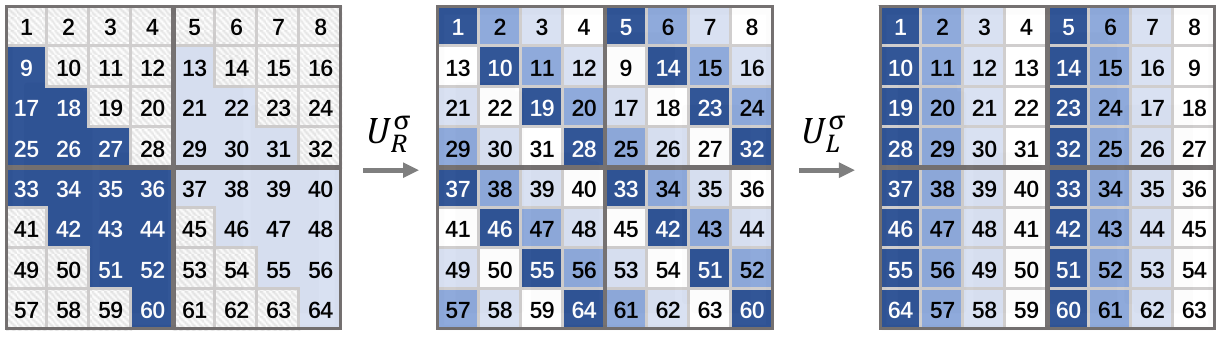}
    \caption{Applying $U_R^\sigma$ and $U_L^\sigma$ to an $8\times 8$ matrix}
    \label{fig:sigma}
\end{figure}

\begin{figure*}[h]
    \centering
        \subfloat[$16\times 16$ $U^\tau$ with non-zero entries painted in dark blue.]{\includegraphics[width=0.5\columnwidth]{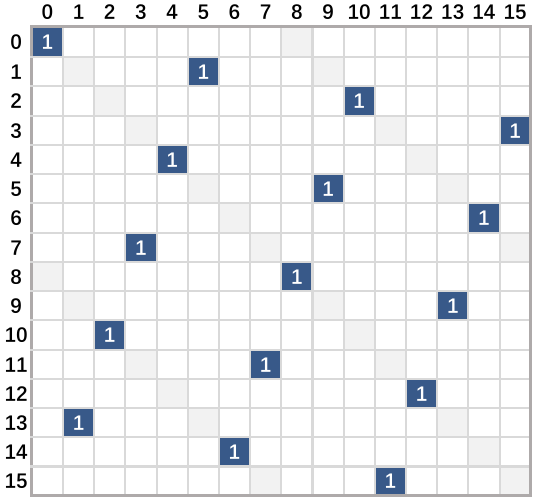}}
        \hfil 
        \subfloat[Step 1. Red entries denote the non-zero entries in $U^\tau_R$ after Step 1.]{\includegraphics[width=0.5\columnwidth]{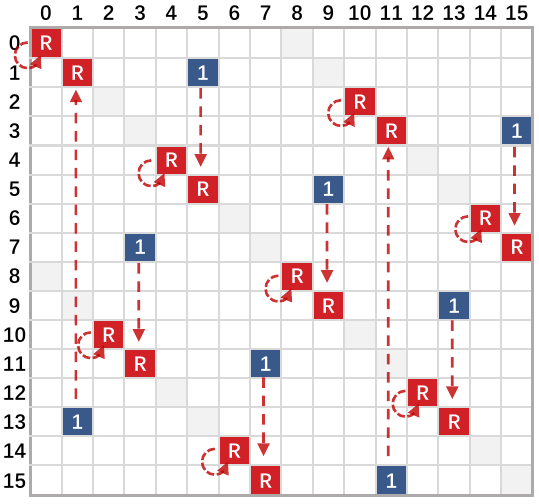}}
        \hfil
        \subfloat[Step 2. No conflicts are generated in the previous step.]{\includegraphics[width=0.5\columnwidth]{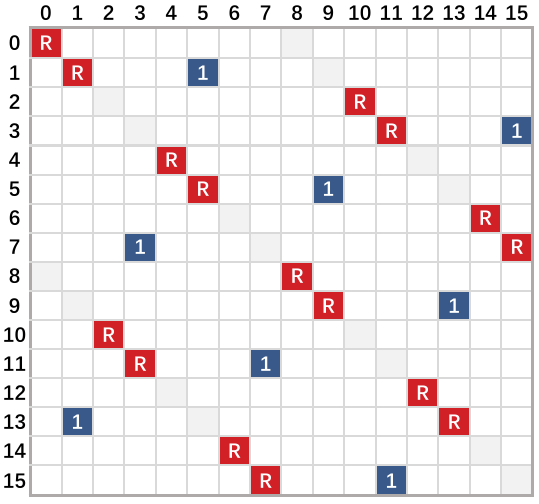}}
        \hfil 
        \subfloat[Step 3. $U^\tau_L$ reconstructed by $U^\tau_R$ and $U^\tau$]{\includegraphics[width=0.5\columnwidth]{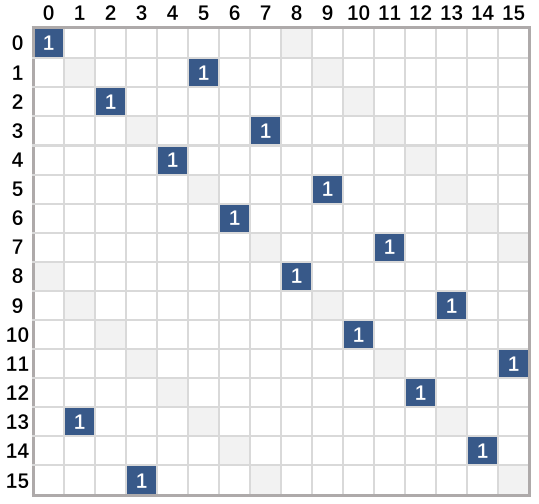}}
    \caption{Applying depth-1 ideal decomposition search to $U^\tau$ of size $16\times 16$}
    \label{fig: depth-1DmpTau}
\end{figure*}

\subsection{Diag-to-Col Transformation}

For any $U^\sigma$ with a power-of-2 dimension, the search algorithm finds one of its depth-$1$ ideal decompositions, denoted as $U^\sigma = U_L^\sigma U_R^\sigma$. Figure \ref{fig: depth-1DmpSigma} shows the decomposition on the $16\times 16$ $U^\sigma$. We can deduce the full-depth decomposability from this depth-$1$ solution. For any $d\times d$ matrix $A$ where $d$ is a power of $2$, the effects of applying $U^\sigma_R$ and $U^\sigma_L$ to $A$ can be described as follows: 
\begin{itemize}[leftmargin=1.1em]
    \item $U^\sigma_R$ evenly partitions $A$ into $4$ blocks of size $d/2\times d/2$ (as indicated in Equation \ref{eq: A_org}), and swaps the lower triangular parts of blocks $A'_{0,0}$ and $A'_{0,1}$, as well as the upper triangular parts of $A'_{1,0}$ and $A'_{1,1}$. 
    \item $U_L^\sigma$ performs a block-wise diag-to-col transformation on $A$. That is: $U^\sigma(A'_{0,0}), U^\sigma(A'_{0,1}), U^\sigma(A'_{1,0}), U^\sigma(A'_{1,1})$.
\end{itemize}
This process, illustrated in Figure \ref{fig:sigma}, can be recursively applied to each block, thus deriving a full-depth ideal decomposition:  
\begin{equation}
U^\sigma(A) = U_{L_\ell}^\sigma\Big( U_{R_\ell}^\sigma \Big( \ldots \Big( U_{R_1}^\sigma(A) \Big) \Big) \Big),
\end{equation}
where \( 1\leq \ell < \log{d} \). For any \( U_{R_i}^\sigma \) (\( 0 < i < \ell \)), it contains only $3$ diagonals, while \( U_{L_\ell}^\sigma \) has only \( d/2^{\ell-1} - 1 \) diagonals (as established by Theorems \ref{thm: 2BS} and \ref{thm: BS4S}).

When $A$'s dimension is not a power of 2, the partition method of Equations \ref{eq: A_divide0} and \ref{eq: A_divide} is employed to handle blocks with odd dimensions. For the overlapping element \( A[d', d'] \) between \( B'_{0,0} \) and \( B'_{1,1} \), \( U^\sigma(B'_{0,0}) \) requires it to shift leftward, while \( U^\sigma(B'_{1,1}) \) imposes no movement constraints on it. Consequently, this overlap causes no conflicts in element movement. Since each round of decomposition generates at most two distinct block sizes, the corresponding \( U_R^\sigma \) in each round involves no more than $5$ non-zero diagonals. 

\subsection{Diag-to-Row Transformation}

For the \( U^\tau \) applying to any \( d \times d \) matrix \( A \), its diagonal distribution is \( [0, d^2 - d] \) with a common difference of \( d \). Such a distribution is not symmetric around the zero diagonal. In this case, we may adjust the objective of the search algorithm: let \( U_R^\tau \) only contain non-zero entries at the $0$-th and $r_c$-th diagonals for $r_c=d \cdot \lceil (d - 1)/2 \rceil$, such that the diagonal distribution of \( U_L^\tau \) converges approximately to \( [0, (d^2 - d)/2] \). 

Figure \ref{fig: depth-1DmpTau} presents the depth-1 decomposition found by the search algorithm. For \( d \) as a power of two, the effects of applying $U_R^\tau$ and $U_L^\tau$ to $A$ can be described as follows:
\begin{itemize}[leftmargin = 1.1em]
    \item \( U^\tau_R \) partitions \( A \) into left and right halves of size $d \times d/2$: \( A = [A'_0, A'_1] \), and rotates the \( A'_1 \) by \( \frac{d}{2} \cdot d \) steps, involving only two diagonals.
    \item \( U_L^\tau \) applies block-wise diag-to-row transformation to $A$: \( U^\tau(A'_0) \) and \( U^\tau(A'_1) \).
\end{itemize}
By recursively applying the partition and permutations mentioned above to each block, a full-depth ideal decomposition is obtained: 
\begin{equation}
U^\tau(A) = U_{L_\ell}^\tau\Big( U_{R_\ell}^\tau \Big( \ldots \Big( U_{R_1}^\tau(A) \Big) \Big) \Big),  
\end{equation}
where \( \ell < \log{d} \). For any \( U_{R_i}^\tau \) (\( 0 < i < \ell \)), it contains only $2$ diagonals, while \( U_{L_\ell}^\tau \) has only \( d/2^\ell \) diagonals.
When \( d \) is not a power of 2, $A$ is divided into $3$ submatrices. For $d=2d'+1$, partition $A$ into \( A = [A'_{0}, A'_{1}, A'_{2}] \), where \( A'_{0} \) and \( A'_{1} \) are \( d \times d' \) submatrices, and \( A'_{2} \) is a \( d \times 1 \) column vector.
\( U_R^\tau \) rotates \( A'_{1} \) by \( d' \cdot d \) steps and \( A'_{2} \) by \( 2d' \cdot d \) steps, resulting in $3$ diagonals. Meanwhile, \( U_L^\tau \) applies the diag-to-row transformation only within the submatrices \( A'_{0} \) and \( A'_{1} \): \( U^\tau(A'_{0}) \) and \( U^\tau(A'_{1}) \).

\section{Section 5 Supplemental Materials}

\subsection{Correctness of the Multi-group Network Construction}\label{app: NWcorrect}

\textcolor{black}{We prove that given any permutation $p$ of length $n$ and any $n$-dimensional input vector $\mathbf{v}$, the proposed multi-group network structure correctly outputs $p(\mathbf{v})$ with at most $\lceil \log{(\max\{r_{org}^{(i)}\})} \rceil$ network levels for $0 \leq i < n$, where $r_{org}^{(i)}$ is the initial \textit{remaining rotation distance} of the $i$-th entry in $\mathbf{v}$ (Recall Section \ref{sec: HVP}).}  

\textcolor{black}{Let us keep in mind that the $i$-th entry of $\mathbf{v}$ will end up being the $i+r_{org}^{(i)}$-th entry of $p(\mathbf{v})$. At any point during the traversal of $\mathbf{v}$'s entries through the network, let $r_{rem}^{(i)}$ represent the $i$-th entry's current remaining rotation distance to its destination, initially set as $r_{org}^{(i)}$. Note that when $r_{rem}^{(i)} = 0$, the entry would have reached its destination.}  
\textcolor{black}{We now prove by induction that for any $0 \leq i < n$, the value $r_{rem}^{(i)}$ will be reduced to $0$ after at most $\lceil \log({\max\{r_{org}^{(i)} \mid 0 \leq i < n\})} \rceil$ levels.}

\textcolor{black}{When entries are initially passed from level $0$ to $1$, we have $r_{rem}^{(i)} = r_{org}^{(i)}$ for $0\leq i <n$. At this level, rotation nodes from all groups ensure that $r_{rem}^{(i)} \leftarrow r_{rem}^{(i)} - rot$ whenever $r_{rem}^{(i)} > rot$, where $rot = 2^{\lfloor \log{m} \rfloor}$ and $m = \max\{r_{org}^{(i)} \mid 0 \leq i < n\}$. Consequently, by the end of level 1, the maximum remaining rotation distance across all entries will have been reduced to less than $\max\{r_{org}^{(i)} \mid 0 \leq i < n\}/2$.}

\textcolor{black}{Assume inductively that upon entering level $k + 1$, the current maximum $r_{rem}^{(i)}$ is less than $\max\{r_{org}^{(i)} \mid 0 \leq i < n\} / 2^k$. Then, as in the base case, the rotation nodes apply a subtraction of $rot = 2^{\lfloor \log{m} \rfloor}$ to any $r_{rem}^{(i)} > rot$, where $m = \max\{r_{rem}^{(i)} \mid 0 \leq i < n\}$. Thus, by the end of level $k+1$, the maximum $r_{rem}^{(i)}$ is again halved, falling below $\max\{r_{org}^{(i)} \mid 0 \leq i < n\} / 2^{k+1}$.}

\textcolor{black}{Therefore, after $\ell = \lceil \log{(\max\{r_{org}^{(i)} \mid 0 \leq i < n\})} \rceil$ levels, we have $\max\{r_{rem}^{(i)} \mid 0 \leq i < n\} < \max\{r_{org}^{(i)}\mid 0\leq i <n\} / 2^\ell$, which equals zero.}

\subsection{Benes network-based Homomorphic Permutation Implementation}\label{app: benes}
The homomorphic permutation based on the Benes network is implemented with the following details: 
\begin{itemize}[leftmargin=1.1em]
    \item The Benes network generated for a permutation of length $n$ typically requires a great multiplication depth of $2\log{n}-1$ which is unlikely to be practical. Thus, we restrict its depth to $\log{n}-1$ using the optimal level-collapsing method proposed by Halevi and Shoup \cite{halevi2014algorithms}. This depth configuration aligns with our new network construction, which is convenient for comparison. The complexity comparison in Table \ref{tb: rotseachLv} is obtained with this setting. With depth $\log{n}-1$ being a bottom line, we further search for the best performance of Benes network-based homomorphic permutation by collapsing more levels. This leads to the results in Table \ref{tb: HVP}. 
    \item The linear transformation at each modulus level is computed using the double-hoisted BSGS algorithm proposed by Bossuat \textit{et al.} \cite{bossuat2021efficient}, which improves efficiency by reorganizing and restructuring sub-modules within the rotations.
    Therefore, the number of rotations in Table \ref{tb: rotseachLv} is obtained through 
    dividing the number of scalar multiplications for all rotation sub-models in the linear transformation by that of a single rotation.
    \item  The number of rotation keys for the Benes network generally grows as the levels collapse. Therefore we fix it to roughly $\log{n}$ for a fair comparison with our network. These keys are selected dynamically to minimize the rotation complexity of the Benes network. 
    

\end{itemize}

\begin{table}
\setlength{\tabcolsep}{4.5pt}
\begin{tabular*}{\columnwidth}{@{}llllllllll@{}}
\toprule
\phantom{Stage}   &  
\multicolumn{1}{c}{HMM by \cite{jiang2018secure}}  &  & 
\multicolumn{3}{c}{Our HMM}  \\
\cmidrule{2-2} \cmidrule{4-6}
Stage       &   Batch 128  &  &  Batch 32 &  Batch 16 & Batch 8\\
\midrule
Convolution     & 121   ms& & 79  ms& 81  ms& 82  ms\\
$1^{st}$ square & 31    ms& & 17  ms& 18  ms& 16  ms\\
FC-1            & 2303  ms& & 788 ms& 379 ms& 233 ms\\
$2^{nd}$ square & 15    ms& & 8   ms& 8   ms& 8   ms\\
FC-2            & 325   ms& & 25  ms& 27  ms& 24  ms\\
Total Latency   & 2795  ms& & 917 ms& 523 ms& 363 ms\\
Amortized       & 22    ms& & 28  ms& 33  ms& 45  ms\\
\bottomrule
\end{tabular*}
\caption{Microbenchmarks for E2DM-lite implemented with different HMM schemes}\label{tb: ONNI1}
\end{table}

\section{Oblivious Neural Network Inference Implementation}\label{app: ONNI}

We demonstrate how to evaluate E2DM-lite homomorphically. Note that E2DM-lite is designed to be compatible with any HMM that supports parallel computation. 

\noindent \textbf{The Convolution Layer.} For a kernel \( K \in \mathbb{R}^{k \times k} \), a stride of \( (s, s) \), \( h \) channels, and an input image \( I \in \mathbb{R}^{n_i \times n_i} \), the convolution operation is given by the formula: 
\[
\textbf{Conv}(I, K)_{i',j'} = \sum_{0 \leq i,j < k} K_{i,j} \cdot I_{si' + i, \, sj' + j}
\]
for \( 0 \leq i', j' < d_K = \left\lceil \frac{(n_i - k)}{s} \right\rceil + 1 \). In E2DM-lite we have $k=7,s=3,h=2,d_K=8$, and $n_i=28$. The encoding method for input of homomorphic convolution is similar to that of E2DM \cite{jiang2018secure}. We first duplicate the set $\{ I_{si' + i, sj' + j} \}_{0 \leq i', j' < d_K}$ \( h \) times to form one column of a matrix $I^{(i,j)}$, which has $b$ such columns for $b$ images to be processed simultaneously. On the other hand, the set $\{K_{i,j}^{(l)}\}_{0\leq l<h}$ is vertically concatenated and then duplicated $d_K^2\times b$ times to form a matrix $K^{(i,j)}$ with the same dimensions as $I^{(i,j)}$.
Encrypting $I^{(i,j)}$ and $K^{(i,j)}$ into a ciphertext, respectively, for $0\leq i,j<k$, the homomorphic evaluation of the convolution layer can be computed as:
\[
ct(C) \leftarrow \sum_{0 \leq i,j < k} \texttt{Mult}(ct(I^{(i,j)}), ct(K^{(i,j)})),
\]
where $C \in \mathbb{R}^{hd_K^2\times b}$ is the output of the layer. Note that the encryption of $I^{(i,j)}$ and $K^{(i,j)}$ follows the encoding rules of the HMM used in the FC-1 layer. 

\noindent \textbf{The First Square Layer.} This step simply squares the input ciphertext by $ct(C^2)\leftarrow \texttt{Mult}(ct(C),ct(C))$. 

\noindent \textbf{The FC-1 Layer.} In the plaintext scenario, this step computes the matrix multiplication of a \( w \times hd_K^2 \) weight matrix and an input matrix of size \( hd_K^2 \times b \) (for E2DM-lite $w=32$). In the ciphertext scenario, this is achieved through HMM across several submatrices. Specifically, when using an HMM of dimension \( d \), we partition \( I^{(i,j)} \) (and \( K^{(i,j)} \)) into \( hd_K^2/d \times b/d \) submatrices, ensuring that \( ct(C^2) \) maintains the same partitioning. We also partition the weight matrix into \( w/d \times hd_K^2/d \) submatrices. This allows us to utilize the parallelism of HMM for \( wb/d^2 \) groups of matrix multiplications, each involving \( d \times d \) submatrices of size \( hd_K^2/d \). The results are then aggregated in a group-wise order to obtain an output $ct(F)$ with $F$ of size \( w \times b \).

\noindent \textbf{The Second Square Layer.} This step squares the input ciphertext by \( ct(F^2) \leftarrow \texttt{Mult}(ct(F), ct(F)) \). 

\noindent \textbf{The FC-2 Layer.} In the plaintext scenario, this step computes the matrix multiplication of an \( o \times w \) weight matrix and an input matrix of size \( w \times b \) (for E2DM-lite $o=10$). Since the partitioning rules are preserved in $ct(F^2)$, this step is similar to the FC-1 layer, performing \( o  b/d^2 \) groups of matrix multiplications, each involving \( d \times d \) submatrices of size \( w/d \). The results are then aggregated to yield an output $ct(O)$ with $O$ of size \( o \times b \). If $d>o$ then zero-padding is applied to the weight matrix. 

Table \ref{tb: ONNI1} presents benchmarks for E2DM-lite, with \( \log{N}=15 \) and an encoding space of length \( 2^{14} \). When using Jiang \textit{et al.}'s HMM, with \( d=32 \) and \( b=128 \), the FC-1 layer performs 4 groups of 4 matrix multiplications of size \( 32 \times 32 \), which, with parallelism support, results in a complexity comparable to performing a single \( 32 \times 32 \) matrix multiplication. When using our HMM, we set three different parameter sets: \( (b=32, d=32, k=4, d_0=8) \), \( (b=16, d=16, k=4, d_0=4) \), and \( (b=8, d=8, k=4, d_0=2) \). This transforms the FC-1 layer into performing 1, 2, and 4 groups of matrix multiplications of size \( 32 \times 32 \), \( 16 \times 16 \), and \( 8 \times 8 \), respectively. The low latency shown in the table is mainly due to the following reasons: (i) our HMM uses $k=4$-fold redundancy encoding space to provide faster computational speed. (ii) The ciphertext output from the FC-1 layer contains \( 2^{14}/(32b) \) duplication of the expected matrix of size $32\times b$, which is sufficient to allow the FC-2 layer to collapse into a single multiplication and a recursive summation with \( \log{32} \) rotations. This enables the entire E2DM-lite computation to be completed at a lower modulus level, where each layer's computational cost is lower than that of the HMM implementation by Jiang \textit{et al.} \cite{jiang2018secure}.

\end{document}